\documentclass[11pt]{article}
\usepackage{fullpage}
\usepackage{amsfonts,amssymb,amsmath,mathtools}
\usepackage[colorlinks=true,breaklinks=true,bookmarks=true,urlcolor=blue,
citecolor=blue,linkcolor=blue,bookmarksopen=false,draft=false]{hyperref}
\usepackage[ruled]{algorithm2e} % for algorithms

\SetAlFnt{\small}
\SetAlCapFnt{\small}
\SetAlCapNameFnt{\small}
\SetAlCapHSkip{0pt}
\IncMargin{-\parindent}
\usepackage{graphicx}
\usepackage[english]{babel}
\usepackage[latin1]{inputenc}
\usepackage{amsthm}
\usepackage{thm-restate}
\usepackage{verbatim}
\usepackage{enumerate}
\usepackage{tikz}
\usetikzlibrary{shapes,snakes}
\usetikzlibrary{backgrounds,calc}
\usepackage{float}
\usepackage{color}
\usepackage{todonotes}
\usepackage{mleftright}
\usepackage{bm}
\usepackage{bbm}
\usepackage{newfloat}
\usepackage{subcaption}
\usepackage{authblk}
\usepackage{anyfontsize}

\DeclareFloatingEnvironment[fileext=lop,name={Algorithm}]{algorithmEnv}

\allowdisplaybreaks

\newtheorem{theorem}{Theorem}
\newtheorem{lemma}{Lemma}

\newtheorem{remark}{Remark}

\newcommand{\revised}[1]{#1} % use this to deactivate blue color for revisions

\newcommand\Prob{\mathbb{P}}
\newcommand\E{\mathbb{E}}

\renewcommand\O{\mathcal{O}}
\newcommand{\erf}[1]{\operatorname{erf}\mleft(#1\mright)}
\newcommand{\e}{\mathrm{e}} % the constant e = 2.71828... (to distinguish it from an edge e)

\begin{document}

\title{Prophet Inequalities over Time}

\date{\today}
\author[1]{Andreas Abels}
\author[2]{Elias Pitschmann}
\author[2]{Daniel Schmand}
\affil[1]{HHU D\"usseldorf, Germany}
\affil[2]{University of Bremen, Germany}

\maketitle

\begin{abstract}
In this paper, we introduce an over-time variant of the well-known prophet inequality with i.i.d.\ random variables. Instead of stopping with one realized value at some point in the process, we decide for each step how long we select the value. Then we cannot select another value until this period is over. The goal is to maximize the expectation of the sum of selected values. We describe the structure of the optimal stopping rule and give upper and lower bounds on the prophet inequality. In online algorithms terminology, this corresponds to bounds on the competitive ratio of an online algorithm. 

We give a surprisingly simple algorithm with a single threshold that results in a prophet inequality of $\approx 0.396$ for all input lengths $n$. Additionally, as our main result, we present a more advanced algorithm resulting in a prophet inequality of $\approx 0.598$ when the number of steps tends to infinity. We complement our results by an upper bound that shows that the best possible prophet inequality is at most $1/\varphi \approx 0.618$, where $\varphi$ denotes the golden ratio.
\end{abstract}

\section{Introduction}
\noindent Prophet inequalities are a well-known concept from optimal stopping theory. A decider, the gambler, samples independent random variables one-by-one. She stops this process at any point and selects the current value. Her goal is to select the largest among all observed and future values. Then the prophet inequality is the ratio between the expected value selected by the decider and the expected value selected by an omniscient prophet that knows the outcome of all the random samples in advance. In the late seventies, Krengel and Sucheston, and Garling~\cite{krengel77} showed that this ratio is $\frac{1}{2}$ even if the variables are distributed differently.

One solution, due to Samuel-Cahn~\cite{samuel-cahn83}, to this problem is to simply select the first sample that beats some pre-computed threshold. Over the last decade, this has sparked renewed interest in prophet inequalities because such thresholds can be interpreted as prices in online mechanisms. Here, random samples represent customers that arrive one by one and are interested in a good. The threshold price is an individual take-it-or-leave-it offer for each customer. Similar models are used for online ad pricing and pricing of cloud services. If nothing is known about the types of customers, it is natural to assume that all customers are identically distributed. In this case, Correa et al.~\cite{DBLP:conf/sigecom/CorreaFHOV17} have shown that the ratio is at least 0.745.

We extend this line of work and, instead of selling goods to customers, we lend them over time. There is a single good and customers arrive online one by one. The value of each customer is an i.i.d.~random variable and the decider determines whether she lends to the customer and for how long. The value collected through the customer is his value times the number of steps that he borrows. Please note that, once we have selected a customer, we cannot cancel this commitment and we have to reject all customers that arrive until the lease expires.

Our approach gives an over-time flavor to prophet inequalities that allows us to apply them to applications with limited, but non-depleting resources. The optimal prophet inequality corresponds to the competitive ratio of an optimal online algorithm for this problem, where the online algorithm describes the decision routine of the decider. Similar to classical prophet inequalities, our algorithms constitute threshold functions that depend on the distribution of the random variables and the number of steps remaining.

Like in other prophet inequalities, our algorithm can also be interpreted as an online posted-price mechanism where the thresholds constitute the prices. In that case, our results directly translate to the social welfare collected under these prices.

\subsection{Related Work}
Classical prophet inequalities go back to Gilbert and Mosteller~\cite{gilbert-mosteller66} and Krengel and Sucheston, and Garling~\cite{krengel77,krengel78}. The latter showed a stopping rule that for $n$ non-negative, independent random variables drawn one after the other, in expectation, collects at least half the maximum value. And they showed that this is best possible. Samuel-Cahn~\cite{samuel-cahn83} simplified the solution and noted that a single threshold suffices for all steps of the sequence.

Interest in prophet inequalities resurged with research on the posted-price mechanism and online mechanism design. Hajiaghayi et al.~\cite{DBLP:conf/aaai/HajiaghayiKS07} and Chawla et al.~\cite{DBLP:conf/stoc/ChawlaHMS10} pointed out the connection and first used stopping rules for prophet inequalities as prices in online mechanisms. Starting from this perspective, a lot of work has been done on different kinds of prophet inequalities, and we can only present a small fraction of this body of work. Kleinberg and Weinberg~\cite{DBLP:journals/geb/KleinbergW19} extended prophet inequalities to matroid and matroid intersection constraints, and D\"utting and Kleinberg~\cite{DBLP:conf/esa/DuttingK15} covered polymatroids. Feldman et al.~\cite{DBLP:conf/soda/FeldmanGL15} analyzed combinatorial auctions with posted-prices. Subsequently, D\"utting et al.~\cite{DBLP:conf/focs/DuttingKL20} gave a prophet inequality in $\mathcal{O}(\log\log n)$ for subadditive combinatorial auctions, and Christi and Correa~\cite{CorreaC23} improved this to a constant. D\"utting et al.~\cite{DBLP:journals/siamcomp/DuttingFKL20} presented a unifying framework for many different types of combinatorial prophet inequalities.

All these results only assume that the number of random variables and their distributions are known. In many applications, it is reasonable to assume that all numbers are drawn i.i.d.\ from the same distribution. Hill and Kertz~\cite{hill-kertz82,kertz86} studied this case and gave a recursive characterization of the prophet inequality for different numbers $n$ of random variables. Correa et al.~\cite{DBLP:conf/sigecom/CorreaFHOV17} showed that this characterization gives tight bounds on the prophet inequality with i.i.d.\ random variables and that its value is approximately 0.745. Subsequently, Correa et al.~\cite{CorreaDFS22} considered the prophet inequality with i.i.d.\ random variables with an unknown distribution. They showed that this prophet inequality is $\frac{1}{e}$, and that this bound holds for a sublinear number of additional samples on the distribution. Additionally, they gave an upper bound of $\O(n^2)$ on the minimum number of samples necessary to gain information equivalent to full knowledge of the distribution. This result has been improved by Correa et al.~\cite{DBLP:journals/corr/abs-2011-06516}. Their analysis is parameterized in the size of the sample compared to the input. For intermediate sample sizes, they achieve a value of $0.745 - \mathcal{O}(\epsilon)$ for $\mathcal{O}(\frac{n}{\epsilon})$ samples.

In contrast to the prophet inequalities discussed above, the work on models that are over time is sparse. For secretary problems, Fiat et al.~\cite{DBLP:conf/esa/FiatGKN15} have coined the temp secretary problem, where each random variable arrives uniformly distributed over a period of time. There is a fixed duration $\lambda$ for all random variables and each random variable has an individual value. Then the decider online selects a subset of the random variables such that no two random variables overlap. Kesselheim and T\"onnis~\cite{DBLP:conf/esa/KesselheimT16} improved and generalized the result and gave a $\frac{1}{2} - \O(\sqrt{\lambda})$-competitive algorithm for the basic problem.
In the prophet inequality minimization domain, Disser et al.~\cite{DBLP:journals/mor/DisserFGGKSST20} showed that, for the minimization version of the problem considered in this paper, it is possible to achieve constant prophet inequalities if the algorithm is required to cover every step of the online process. This means, there is no restriction on the number of values selected simultaneously. In addition, they showed that this is not possible if the decider is restricted to selecting only a single value at a time. This result is in contrast to the impossibility result by Esfandiari et al.~\cite{DBLP:journals/siamdm/EsfandiariHLM17} who showed that even for i.i.d.\ random variables, there is no constant prophet inequality for the minimization version. For a maximization over time similar to the model in this paper, where the lending durations are not a decision of the algorithm but instead a part of the online input, Feng et al.~\cite{feng2022near} showed a lower bound on the competitive ratio of $0.5$. Faw et al.~\cite{faw2022learning} refine this result for the special case where the lending duration is one fixed parameter $d\in o(n)$. In this setting, they prove a tight competitive ratio of $d/(2d-1)>0.5$.

Recently, Cristi and Oren~\cite{cristi2024planning} introduced a graph-theoretic framework for prophet inequalities, where the algorithm traverses a graph and the available decisions depend on the algorithm's position in the graph. Our model is a special case of this setting and we obtain a lower bound of $0.5$ on our prophet inequality by~\cite{cristi2024planning}. In this work, we improve the lower bound to $0.598$ for $n\to\infty$.

\revised{Independent of our work Perez-Salazar and Verdugo~\cite{perez2024optimal} presented a single-threshold algorithm for prophet inequalities over time with a competitive ratio of $0.567$ and proved that no single-threshold algorithm can outperform this guarantee. Moreover, they improved the general lower bound for the problem to $0.618$, matching our upper bound in Section \ref{sec:upperbound}, by modeling the optimal dynamic programming approach for a worst-case instance though convex programming. This shows the tightness of our upper bound. Nevertheless, our algorithm presented in Section \ref{sec:betterlower} obtains a near-optimal guarantee on the competitive ratio while using significantly easier thresholds given in a closed form.}

\subsection{Our Contribution}
We introduce an over-time component to prophet inequalities with i.i.d.\ distributed random variables. Instead of stopping at some point, we decide for how many steps we want to select a value, and then we cannot select another one until that period is over. The objective is to maximize the expectation of the sum of values taken in each step.

It is easy to see that, as the distribution of the random variables is known, the optimal algorithm for this problem results from a simple dynamic program. Like with ordinary prophet inequalities, it is not clear how to compute the value of this prophet inequality, the competitive ratio of this optimal algorithm. Nevertheless, we show that the optimal online algorithm has a simple form. When it has no value selected for the current step, it selects the current value either for a single step or until the end of the input sequence. Based on this paradigm, we present a multi-threshold rule, where the current threshold is a quantile that changes over the input sequence. We bound this prophet inequality in the limit and show that it is $0.598$ when $n$, the length of the input sequence, tends to infinity. We complement this result with an upper bound of $\frac{1}{\varphi} \approx 0.618$ where $\varphi$ is the golden ratio. This is a strict separation from classical prophet inequalities with i.i.d.\ random variables where a prophet inequality of value 0.745 exists~\cite{DBLP:conf/sigecom/CorreaFHOV17}.  

We present the algorithmic result in an incremental way. As a warm-up, we give a simple single-threshold algorithm where the threshold is a specific quantile of the distribution that only depends on the total length of the input sequence. By using similar techniques as in the analysis of the multi-threshold rule, we show that this simple algorithm results in a prophet inequality of $\frac{2+e^{-2}}{3-e^{-2}} \approx 0.396$ for \emph{all lengths} of the input sequence $n$. Subsequently, we extend these techniques to the analysis of the multi-threshold rule.

\section{The Model}\label{sec:Model}
For a natural number $n$, let $[n] = \{1,\dots,n\}$. Consider a sequence of $n>0$ i.i.d.\ random variables $X_1,\dots,X_n$, from a common distribution. We denote the corresponding cumulative distribution function as $F$. In technical parts of the paper, we will usually assume continuity of $F$, which is without loss of generality due to the arguments in the last paragraph of this section. Further, we assume that the distribution assigns positive probability to non-negative numbers only, that is $F(0)=0$. In every step $i\in[n]$, the task is to choose a number of steps $t_i\in[n-i+1]$ for the value in step $i$ to be \textit{selected}. We write $x_i$ for the realization of $X_i$ and gain a value of $x_i$ for $t_i$ steps. In case we already have some other value selected in step $i$, we have to discard the random variable and set $t_i=0$. The objective is to maximize the expected overall gained value $\E\left[\sum_{i=1}^n X_i t_i\right]$. Throughout the paper, we often write $X$ or $x$ without index in case the step $i$ is either clear from the context or the statement is general and holds for all $i$.

A \textit{decision procedure} for this problem can also be seen as an online algorithm, which we denote by \textsc{Alg}. It has access to the number of steps $n$, the cumulative distribution function $F$ of the random variables, and in step $i$ to the realizations $x_1,\dots,x_i$ as well as its own decisions $t_1,\dots,t_{i-1}$. The prophet inequality compares the expected value obtained by a gambler, i.e., the online algorithm \textsc{Alg} to a prophet, an omniscient \textit{offline algorithm} \textsc{Opt}. The prophet sees all the realizations of the random variables beforehand and makes optimal decisions that maximize the objective function. The competitive analysis of online algorithms corresponds to the prophet inequalities in this setting.  For any given algorithm \textsc{Alg}, we denote the expected overall value it achieves on an input of length $n$ as $\E[\textsc{Alg}_n]$. Note that we drop the dependence on $F$ if it is clear from context. The ratio $\sup_{F,\textsc{Alg}^\prime} (\E[\textsc{Alg}^\prime_n] / \E[\textsc{Opt}_n])$ denotes the value of the prophet inequality over time.

Our algorithms are based on thresholds on the value that is realized in a step. For these thresholds, we use quantiles of the distribution. Throughout the paper, we write $\delta_p$ for the \textit{$p$-quantile} of the given distribution, i.e., the value such that $\Prob_{X \sim F}[X \leq \delta_p] = p$. In case $p$ is negative, we assume $\delta_p = 0$. Note that the quantiles according to this definition do exist for continuous probability distributions. For a discrete probability distribution, a slight perturbation of the drawn values yields a continuous probability distribution as follows. Suppose we are given an instance with a discrete probability distribution and known $n$. For a corresponding discrete value $x$ with positive probability $\alpha$ in the probability distribution, one can introduce a continuous interval $x'=[x-\varepsilon,x+\varepsilon]$ for tiny $\varepsilon>0$ with $Pr[X \in x'] = \alpha$. This induces a total additional error of at most $n\varepsilon$ in the objective function. By choosing $\varepsilon$ small enough (in dependence on $n$ and the expected value of \textsc{Opt}) the change in the value of the prophet inequality is negligible.

\section{Structure of an Optimal Decision Procedure}
\label{structureOnline}
We derive and discuss a short description of an optimal decision procedure. Let $G_n$ denote the expected value achieved by an optimal procedure for $n$ steps. Define $\tau_{n-i} \coloneqq G_{n-i}/(n-i)$ for $i=0,\dots,n-1$ and $\tau_0 = 0$, i.e., $\tau_{n-i}$ denotes the expected average value per step of an optimal procedure for $n-i$ random variables. Consider the following algorithm.

\begin{algorithmEnv}[ht]
  \begin{algorithm}[H]
\For{$i=1,\dots,n$}{
\eIf{$x_i > \tau_{n-i}$}{
select $x_i$ for all $n-i+1$ remaining steps;\\
\textbf{break}
}{select $x_i$ only in step $i$;}}
\caption{An Optimal Decision Procedure for Prophet Inequalities over Time}
\label{alg:onlineopt}
\end{algorithm}
\caption{An Optimal Decision Procedure for Prophet Inequalities over Time.}
\end{algorithmEnv}

Interestingly, for a given realization, the procedure either selects the value for all remaining steps or a single step. We show that Algorithm \ref{alg:onlineopt} is indeed optimal among all possible online procedures.
\begin{theorem}
    Algorithm \ref{alg:onlineopt} is an optimal online algorithm for prophet inequalities over time.
    \label{thm:structureOptimalOnlineAlgo}
\end{theorem}

\begin{proof}
We provide a proof by induction over $n$ that a) Algorithm \ref{alg:onlineopt} is optimal for a sequence of $n$ random variables and b) $\tau_i\geq\tau_{i-1}$ for all $1 \leq i \leq n$.\\
For $n=1$, both claims are obvious.\\
Now assume both claims to be true for all numbers of steps $1,\dots n-1$. First, we prove statement a) for $n$ steps. The best possible decision procedure for the prophet inequality over time problem can easily be characterized through a dynamic program. Clearly, $G_0 = 0$ and
\begin{align}\label{eq:recursiveonline}
    G_n = \E_{X \sim F}\left[\max_{t\in[n]}\{ t  x + G_{n-t}\} \right].
\end{align}
We distinguish two cases. First, consider $x > \tau_{n-1}$. We have to show that the maximum (\ref{eq:recursiveonline}) is attained for $t=n$. By induction hypothesis, $\tau_{n-1}\geq\dots\geq\tau_0$, so for all $t=1,\dots,n-1$,
\[ n x + G_{0} > t x + (n-t)\tau_{n-1} \geq t x + (n-t)\tau_{n-t} = t x + G_{n-t}. \]
Next, consider $x \leq \tau_{n-1}$. We have to show that the maximum (\ref{eq:recursiveonline}) is attained for $t=1$. By induction hypothesis, $\tau_{n-1}\geq\dots\geq\tau_0$, so for all $t=1,\dots,n$,
\[ x + G_{n-1} = x + (t-1)\tau_{n-1}+ (n-t)\tau_{n-1} \geq x + (t-1)x + (n-t)\tau_{n-t} = t x + G_{n-t}. \]
It remains to show $\tau_n \geq \tau_{n-1}$. From the above, we have
\[ G_n = \Prob[X > \tau_{n-1}] n \E[X\mid X > \tau_{n-1}] + \Prob[X \leq \tau_{n-1}] (\E[X\mid X \leq \tau_{n-1}] + G_{n-1}). \]
Plugging this into the definition of $\tau_n$ and using the law of total probability yields
\begin{align}\label{eq:analogoustau}
\begin{split}
    \tau_n &= \Prob[X > \tau_{n-1}] \E[X\mid X > \tau_{n-1}] + \frac{1}{n}\Prob[X \leq \tau_{n-1}] (\E[X\mid X \leq \tau_{n-1}] + G_{n-1})\\
    &= \E[X] - \frac{n-1}{n} \Prob[X \leq \tau_{n-1}] (\E[X\mid X \leq \tau_{n-1}] - \tau_{n-1})\\
    &\geq \E[X] - \frac{n-2}{n-1} \Prob[X \leq \tau_{n-1}] (\E[X\mid X \leq \tau_{n-1}] - \tau_{n-1}),
\end{split}
\end{align}
where the last inequality holds because $\E[X\mid X \leq \tau_{n-1}] - \tau_{n-1} \leq 0$. Applying Lemma \ref{lem:onlineopt} with $\tau_{n-1} \geq \tau_{n-2}$ by induction hypothesis, we finally find that
\[ \tau_n \geq \E[X] - \frac{n-2}{n-1} \Prob[X \leq \tau_{n-2}] (\E[X\mid X \leq \tau_{n-2}] - \tau_{n-2}) = \tau_{n-1},\]
where the latter equality can be derived with an analogous calculation to Equation \eqref{eq:analogoustau}.
\end{proof}

\begin{lemma}\label{lem:onlineopt}
    The function $g(\tau) \coloneqq \Prob[X \leq \tau] (\E[X\mid X \leq \tau] - \tau)$ is non-increasing.
\end{lemma}

\begin{proof}
    We verify for $\tau^\prime > \tau$ that
    \begin{align*}
        g(\tau^\prime) - g(\tau) &= \E[\mathbbm{1}_{X\leq\tau^\prime} X] - \tau^\prime \Prob[X \leq \tau^\prime] - \E[\mathbbm{1}_{X\leq\tau} X] + \tau \Prob[X \leq \tau] \\
        &= \E[(\mathbbm{1}_{X\leq\tau^\prime} - \mathbbm{1}_{X\leq\tau}) X] - \tau^\prime \Prob[X \leq \tau^\prime] + \tau \Prob[X \leq \tau] \\
        &= \E[\mathbbm{1}_{\tau \leq X \leq\tau^\prime} X] - \tau^\prime \Prob[X \leq \tau^\prime] + \tau \Prob[X \leq \tau] \\
        &\leq \tau^\prime (\Prob[X \leq \tau^\prime] - \Prob[X \leq \tau]) - \tau^\prime \Prob[X \leq \tau^\prime] + \tau \Prob[X \leq \tau]\\
        &= (\tau - \tau^\prime) \Prob[X \leq \tau]\\
        &\leq 0\;,
    \end{align*}
    which finishes the proof.
\end{proof}

Interestingly, the structure of Algorithm $\ref{alg:onlineopt}$ is very easy and its thresholds directly come from the dynamic program for $G_n$. Unfortunately, it is not clear how to use the recursive structure of the thresholds to derive a lower bound on the obtained prophet inequality. At the same time, we cannot derive a usable closed form of the thresholds to work with. Consequently, we will describe other procedures that perform well and we are able to analyze.

\section{A Simple Lower Bound on the Prophet Inequality over Time}
We derive a first lower bound on the prophet inequality over time through the analysis of a surprisingly simple online algorithm, which we call \textsc{Simple}, using only one threshold $\delta_{1-\frac{a}{n}}$ for some fixed value $a\in(0,n)$ for a given number of steps $n\in\mathbb{N}_{>0}$. Note that this procedure only uses a single threshold that only depends on the distribution and on the number of drawn random variables $n$ and stays constant throughout the whole procedure.

\begin{algorithmEnv}[ht]
  \begin{algorithm}[H]\caption{\textsc{Simple}}\label{alg:simple}
\For{$i=1,\dots,n$}{
\eIf{$x_i > \delta_{1-\frac{a}{n}}$}{
select $x_i$ for all $n-i+1$ remaining steps;\\
\textbf{break}
}{select $x_i$ only in step $i$;}}
\end{algorithm}
  \caption{A Simple Decision Procedure for Prophet Inequalities over Time with Only One Threshold}
\end{algorithmEnv}

\begin{theorem}\label{thm:simple}
    The \textsc{Simple} algorithm with threshold $\delta_{1-\frac{2}{n}}$ provides a lower bound of $\frac{1 + \e^{-2}}{3 - \e^{-2}} > 0.396$ for the prophet inequality over time for any $n \in \mathbb{N}_{>0}$.
\end{theorem}
To prove the theorem, we will show a lower bound on the expected value of the \textsc{Simple} algorithm and an upper bound on the expected value of the omniscient optimal offline algorithm \textsc{Opt} that knows all realizations in advance. We will parameterize the proof for \textsc{Simple} in $a \in (0,n)$ to optimize over $a$ afterwards.

\begin{lemma}\label{lem:simplelow}
    The expected value of the \textsc{Simple} algorithm for $n$ steps with thre\-shold $\delta_{1-\frac{a}{n}}$ is lower bounded by
    \[\E[\textsc{Simple}_n] \geq  \left( n \left( 1 - \frac{1}{a} + \frac{1}{a} \e^{-a}  \right) + 1 - (a+2)\e^{-a} \right) \E[X \mid X > \delta_{1-\frac{a}{n}}].\]
\end{lemma}

\begin{proof}
Assume the first realized value that is greater than $\delta_{1-\frac{a}{n}}$ appears in step $i\in[n]$. Then, the algorithm selects this value $x_i$ for the current and the remaining $(n-i)$ steps. The expected realized value is $\E[X \mid X > \delta_{1-\frac{a}{n}}]$. For the simplicity of this analysis, we ignore all values from steps before step $i$ here. We have,
    \begin{align*}
        &\E[\textsc{Simple}_n] \geq \sum_{i=1}^n \left(1-\frac{a}{n}\right)^{i-1} \frac{a}{n} (n-i+1) \E[X_i \mid X > \delta_{1-\frac{a}{n}}]\\
        &= \left( a\frac{n+1}{n} \sum_{i=1}^n \left(1-\frac{a}{n}\right)^{i-1} - \frac{a}{n} \sum_{i=1}^n i \left(1-\frac{a}{n}\right)^{i-1} \right) \E[X\mid X > \delta_{1-\frac{a}{n}}]\\
        &= \left( a\frac{n+1}{n} \sum_{i=0}^{n-1} \left(1-\frac{a}{n}\right)^i - \frac{a}{n} \frac{n \left(1-\frac{a}{n}\right)^{n+1} - (n+1) \left(1-\frac{a}{n}\right)^n + 1}{\left(1-\left(1-\frac{a}{n}\right)\right)^2} \right) \E[X\mid X > \delta_{1-\frac{a}{n}}]\\
        &= \left( a\frac{n+1}{n} \frac{1-\left(1-\frac{a}{n}\right)^n}{1-\left(1-\frac{a}{n}\right)} - \frac{n}{a}\left(1-\frac{a}{n}\right)^n \left( (n-a) - (n+1) \right) - \frac{n}{a} \right) \E[X\mid X > \delta_{1-\frac{a}{n}}]\\
        &= \left( (n+1) \left(1-\smash[b]{\underbrace{\left(1-\frac{a}{n}\right)^n}_{\leq \e^{-a}}}\right) + \left(\frac{n}{a}-1\right) \smash[b]{\underbrace{\left(1-\frac{a}{n}\right)^{n-1}}_{\geq \e^{-a}}} (a+1) - \frac{n}{a} \right) \E[X\mid X > \delta_{1-\frac{a}{n}}]\\[0.5cm]
        &\geq \left( n - n\e^{-a} + 1 - \e^{-a} + \frac{n(a+1)}{a}\e^{-a} -(a+1)\e^{-a} -\frac{n}{a} \right) \E[X\mid X > \delta_{1-\frac{a}{n}}]\\
        &= \left( n \left( 1 - \frac{1}{a} + \frac{1}{a} \e^{-a}  \right) + 1 - (a+2)\e^{-a}\right) \E[X\mid X > \delta_{1-\frac{a}{n}}]\;,
    \end{align*}
which finishes the proof.
\end{proof}

Next, we upper bound the value of the optimal offline algorithm.
\begin{lemma}\label{lem:simpleupp}
    For arbitrary $a\in(0,n)$, the expected value of the optimal offline algorithm \textsc{Opt} for $n$ steps is upper bounded by
    \[\E[\textsc{Opt}_n] \leq \left( n \left( \frac{a}{2} + \frac{1}{a} - \frac{1}{a}\e^{-a} \right) + \frac{a}{2} + 2\e^{-a} - 1 \right) \E[X\mid X > \delta_{1-\frac{a}{n}}].\]
\end{lemma}

\begin{proof}
    The optimal offline algorithm has the so-far best-realized value selected in each step, that is \[\E[\textsc{Opt}_n] = \E\left[\sum_{i=1}^n \max\{X_1,\dots,X_i\}\right] = \sum_{i=1}^n \E[\max\{X_1,\dots,X_i\}]\;.\]
    We provide an upper bound for each summand separately. In case none of the first $i$ random variables realizes above $\delta_{1-\frac{a}{n}}$, we can upper bound the expected maximum by $\E[X\mid X > \delta_{1-\frac{a}{n}}]$. In case exactly $j \geq 1$ of the first $i$ random variables have a realization greater than $\delta_{1-\frac{a}{n}}$, the expected maximum is at most $j \E[X\mid X > \delta_{1-\frac{a}{n}}]$ due to the fact that $\E[\max\{X_1,\dots,X_j\}] \leq j \E[X_1]$ for i.i.d.\ random variables $X_1,\dots,X_j$. Formally, for all $i\in[n]$,
    \begin{align*}
        &\E[\max\{X_1,\dots,X_i\}] \leq \left(1 - \frac{a}{n}\right)^i \E[X\mid X > \delta_{1-\frac{a}{n}}] + \sum_{j=1}^i \binom{i}{j} \left(\frac{a}{n}\right)^j \left(1-\frac{a}{n}\right)^{i-j} j \E[X\mid X > \delta_{1-\frac{a}{n}}],
    \end{align*}
    where the latter sum equals the expectation of a $(i,a/n)$-distributed binomial variable, which is $i a/n$. Thus,
    \begin{align*}
        &\E[\textsc{Opt}_n] \leq \sum_{i=1}^n \left( \left(1 - \frac{a}{n}\right)^i + i\frac{a}{n} \right) \E[X\mid X > \delta_{1-\frac{a}{n}}]\\
        &= \left( \frac{\left(1 - \frac{a}{n}\right) (1-\left(1 - \frac{a}{n}\right)^n)}{1-\left(1 - \frac{a}{n}\right)} + \frac{a}{n}\frac{n(n+1)}{2} \right) \E[X\mid X > \delta_{1-\frac{a}{n}}]\\
        &= \left( \left(\frac{n}{a} - 1\right) \left(1-\left(1 - \frac{a}{n}\right) \smash[b]{\underbrace{\left(1 - \frac{a}{n}\right)^{n-1}}_{\geq \e^{-a}}}\right) + (n+1)\frac{a}{2} \right) \E[X\mid X > \delta_{1-\frac{a}{n}}]\\[0.5cm]
        &\leq \left( \frac{n}{a} - \frac{n}{a}\e^{-a} + \e^{-a} - 1 + \e^{-a} - \frac{a}{n}\e^{-a} + n\frac{a}{2} + \frac{a}{2} \right) \E[X\mid X > \delta_{1-\frac{a}{n}}]\\
        &= \left( n \left( \frac{a}{2} + \frac{1}{a} - \frac{1}{a}\e^{-a} \right) + \frac{a}{2} + 2\e^{-a} - 1 - \frac{a}{n}\e^{-a} \right) \E[X\mid X > \delta_{1-\frac{a}{n}}]\\
        &\leq \left( n \left( \frac{a}{2} + \frac{1}{a} - \frac{1}{a}\e^{-a} \right) + \frac{a}{2} + 2\e^{-a} - 1 \right) \E[X\mid X > \delta_{1-\frac{a}{n}}]\;
    \end{align*}
    which finishes the proof.
\end{proof}

Combining the two previous lemmas, we obtain a lower bound on the prophet inequality of \textsc{Simple}.
\[
    \frac{\E[\textsc{Simple}_n]}{\E[\textsc{Opt}_n]} \geq \frac{n \left( 1 - \frac{1}{a} + \frac{1}{a} \e^{-a}  \right) + 1 - (a+2)\e^{-a}}{n \left( \frac{a}{2} + \frac{1}{a} - \frac{1}{a}\e^{-a} \right) + \frac{a}{2} + 2\e^{-a} - 1}
\]

A numerical optimization over $a$ yields a best possible lower bound obtainable by this approach of $0.3965$ when choosing $a \approx 2.083$. For simplicity, we choose $a=2$ for Theorem \ref{thm:simple}.

\begin{proof}[Proof of Theorem \ref{thm:simple}]
    First, note that for $n=1$, every procedure has to select the single value for one step, which yields a prophet inequality of $1$.\\
    Next, for $a = 2$ and general value of $n\in\mathbb{N}_{>1}$, the corresponding fraction of the bounds from the lemmas is non-increasing in $n$:\[
    \frac{\E[\textsc{Simple}_n]}{\E[\textsc{Opt}_n]} \geq \frac{n \left(\frac{1}{2} + \frac{1}{2} \e^{-2} \right) + 1 - 4\e^{-2}}{n \left( \frac{3}{2} - \frac{1}{2}\e^{-2} \right) + 2\e^{-2}} \geq \frac{(n+1) \left(\frac{1}{2} + \frac{1}{2} \e^{-2}\right) + 1 - 4\e^{-2}}{(n+1) \left( \frac{3}{2} - \frac{1}{2}\e^{-2}\right) + 2\e^{-2}},\]
    where the latter inequality can be shown by multiplying the numerator of one fraction with the denominator of the other and comparing the terms. Consequently, the fraction is lower bounded for all $n\in\mathbb{N}_{>0}$ by 
    \[\frac{\E[\textsc{Simple}_n]}{\E[\textsc{Opt}_n]} \geq \lim_{n\to\infty} \frac{n \left(\frac{1}{2} + \frac{1}{2} \e^{-2} \right) + 1 - 4\e^{-2}}{n \left( \frac{3}{2} - \frac{1}{2}\e^{-2} \right) + 2\e^{-2}} = \frac{\frac{1}{2} + \frac{1}{2} \e^{-2}}{\frac{3}{2} - \frac{1}{2}\e^{-2}} = \frac{1 + \e^{-2}}{3 - \e^{-2}}\approx 0.3963\;,\]
    which finishes the proof.
\end{proof}

%%%%%%%%%%%
% SECTION %
%%%%%%%%%%%

\section{Upper Bound on the Prophet Inequality over Time}\label{sec:upperbound}

In this section, we provide an upper bound on the best possible prophet inequality that is achievable by any online algorithm. We show the following theorem.

\begin{theorem}
The prophet inequality over time is at most $1/\varphi \approx 0.618$, where $\varphi = \frac{1 + \sqrt{5}}{2}$ denotes the golden ratio.
\label{thm:lower}
\end{theorem}

To prove the statement of the theorem, we construct a specific example instance, i.e., a distribution, where the behavior of a best possible procedure is easy to analyze. We assume the input length $n$ to be large and choose the distribution such that the random variable $X$ is given by
\[X = \begin{cases}\varphi n & \text{with prob.\ }\frac{1}{n^2}\;,\\
1 & \text{with prob.\ }\frac{1}{\sqrt{n}}\;,\\
0 & \text{with prob.\ }1 - \frac{1}{\sqrt{n}} - \frac{1}{n^2}\;.\end{cases}\]

Note that this is a discrete probability distribution. However, a slight perturbation of the drawn values yields a continuous probability distribution. Let $F$ denote its cumulative distribution function.

 We derive a closed form for the expected profit of a best possible online algorithm. Let $\E[\textsc{Alg}_n]$ denote this expectation for an instance of $n$ steps with random variables distributed according to $F$. Obviously, $\E[\textsc{Alg}_1]$ is given by
\[\E[\textsc{Alg}_1] = \E[X] = \varphi n\cdot \frac{1}{n^2} + 1 \cdot \frac{1}{\sqrt{n}} = \frac{\varphi}{n} + \frac{1}{\sqrt{n}}\;.\]

In Section \ref{structureOnline}, we have shown that an optimal algorithm can be described by a dynamic program. The optimal decision in step $n-k-1$ is to select the current realization for all remaining steps if and only if the value is at least $\E[\textsc{Alg}_k]/k$. Otherwise, the value is only selected for the current step. 

We observe that an optimal algorithm has the following behavior on the constructed instance: If a random variable realizes at $\varphi n$, we select the value for all remaining steps and stop. If it realizes at $0$, we reject it. If it realizes at $1$, we select it for all remaining steps if and only if $1> \E[\textsc{Alg}_k] / k$ when $k$ steps follow.

Recall that $\E[\textsc{Alg}_k] / k$ is monotonically increasing by Theorem \ref{thm:structureOptimalOnlineAlgo}. Thus for sufficiently large $n$, there is a step $k^\prime$ for which $\E[\textsc{Alg}_{k^\prime}] / k^\prime \leq 1$ and $\E[\textsc{Alg}_k] / k > 1$ for all $k>k^\prime$. For our bound, we derive $k^\prime$ for the given distribution $F$. To do so, we use the fact that we know the exact behavior of the algorithm and start by showing a closed form for $\E[\textsc{Alg}_k]$ for all $k \leq k^\prime$ by induction over $k$.% The proof of the corresponding Lemma \ref{lem:lowersmallk} can be found in the appendix.

\begin{restatable}{lemma}{lowersmallk}
Let $\E[\textsc{Alg}_k]$ denote the expected value of an optimal algorithm for the instance with $k$ steps and cumulative distribution function $F$ defined as above. For $k \leq k^\prime$, we can write
    \[\E[\textsc{Alg}_k] = \frac{\varphi n + n \sqrt{n}}{1 + n\sqrt{n}}{\left(k + \left(\left(1- \frac{1}{n^2} - \frac{1}{\sqrt{n}}\right)^k -1 \right) \frac{n^2 - n\sqrt{n} - 1}{1 + n\sqrt{n}}\right)}\,.\]
    \label{lem:lowersmallk}
\end{restatable}

\begin{proof}
    As $k \leq k^\prime$, we know the exact behavior of an optimal online algorithm. A random variable with value 0 will be selected for a single step and a random variable with strictly positive value will be selected for all remaining steps. We prove the claimed formula by an induction over $k$.\\
    
    For the base case, let $k=1$. We observe
    \[\E[\textsc{Alg}_1] = \frac{1}{n^2} \cdot \varphi n + \frac{1}{\sqrt{n}} = \frac{\varphi}{n} + \frac{1}{\sqrt{n}}.\]
    \revised{Checking the formula in the lemma for $k=1$ yields}
    \begin{align*}
  & \frac{\varphi n + n \sqrt{n}}{1 + n\sqrt{n}}{\left(1 + \left(\left(1- \frac{1}{n^2} - \frac{1}{\sqrt{n}}\right) -1 \right) \frac{n^2 - n\sqrt{n} - 1}{1 + n\sqrt{n}}\right)}\\
        &= \frac{\varphi n + n \sqrt{n}}{1 + n\sqrt{n}} + \left(\frac{\varphi n + n \sqrt{n}}{1 + n\sqrt{n}}\right) \left(\frac{-1-n\sqrt{n}}{n^2}\right) \left(\frac{n^2 - n\sqrt{n} - 1}{1 + n\sqrt{n}}\right)\\
        &= \frac{\varphi n + n \sqrt{n}}{1 + n\sqrt{n}} - \frac{(\varphi n + n \sqrt{n})(n^2 - n\sqrt{n}-1)}{n^2 (1 + n\sqrt{n})}\\
        &= \frac{\varphi n + n \sqrt{n}-\varphi n + \varphi \sqrt{n} + \frac{\varphi}{n} - n\sqrt{n} + n + \frac{1}{\sqrt{n}}}{1 + n\sqrt{n}}\\
        &= \frac{\varphi}{n} \left(\frac{n \sqrt{n} +1}{1 + n \sqrt{n}}\right) + \frac{1}{\sqrt{n}} \left(\frac{n \sqrt{n} +1}{1 + n \sqrt{n}}\right) \revised{= \E[\textsc{Alg}_1]}
\end{align*}
    which finishes the base case. We continue by assuming the formula for $k-1$ and proving it for $k\leq\min\{n,k^\prime\}$. We get
    \begin{align*}
        &\E[\textsc{Alg}_k] = \frac{1}{n^2} \cdot \varphi n \cdot k + \frac{1}{\sqrt{n}} k + \left(1-\frac{1}{\sqrt{n}} - \frac{1}{n^2}\right) \E[\textsc{Alg}_{k-1}]\\
        &= \frac{1}{n^2} \cdot \varphi n \cdot k + \frac{1}{\sqrt{n}} k\\
        & + \left(1-\frac{1}{\sqrt{n}} - \frac{1}{n^2}\right) \frac{\varphi n + n \sqrt{n}}{1 + n\sqrt{n}}{\left(k-1 + \left(\left(1- \frac{1}{n^2} - \frac{1}{\sqrt{n}}\right)^{k-1} -1 \right) \frac{n^2 - n\sqrt{n} - 1}{1 + n\sqrt{n}}\right)}  \\
        &= k\left(\frac{\varphi n + n \sqrt{n}}{n^2} + \frac{n^2-n\sqrt{n}-1}{n^2}\frac{\varphi n + n\sqrt{n}}{1 + n\sqrt{n}}\right)\\
        & + \frac{\varphi n + n\sqrt{n}}{1 + n\sqrt{n}} \left(- \left(1-\frac{1}{\sqrt{n}} - \frac{1}{n^2}\right) + \left(1-\frac{1}{\sqrt{n}}-\frac{1}{n^2}\right) \bigg(\bigg(1-\frac{1}{n^2}-\frac{1}{\sqrt{n}}\bigg)^{k-1}-1\bigg)\frac{n^2 - n\sqrt{n} - 1}{1 + n\sqrt{n}}\right)\\
        &= k \frac{\varphi n + n\sqrt{n}}{1 + n\sqrt{n}} \left(\frac{1+n\sqrt{n}}{n^2} + \frac{n^2-n\sqrt{n}-1}{n^2}\right)\\
        & + \frac{\varphi n + n\sqrt{n}}{1 + n\sqrt{n}}\; \frac{n^2 - n\sqrt{n} - 1}{1 + n\sqrt{n}} \left(- \frac{1+n\sqrt{n}}{n^2} + \left(1-\frac{1}{n^2}-\frac{1}{\sqrt{n}}\right)^{k} - \frac{n^2-n\sqrt{n} - 1}{n^2}\right)\\
        &= k \frac{\varphi n + n\sqrt{n}}{1 + n\sqrt{n}} +  \frac{\varphi n + n\sqrt{n}}{1 + n\sqrt{n}}\; \frac{n^2 - n\sqrt{n} - 1}{1 + n\sqrt{n}} \left(\left(1- \frac{1}{n^2} - \frac{1}{\sqrt{n}}\right)^k -1 \right)
    \end{align*}
    which finishes the proof.
\end{proof}

In order to calculate $k^\prime$, we aim to find the largest $k$ such that $\E[\textsc{Alg}_k] \leq k$. For a fixed $n$, we denote $k = a n$ for some $a \in [0,1]$. By Lemma~\ref{lem:lowersmallk}, we obtain
\begin{align*}
   \E[\textsc{Alg}_{an}] &= \frac{\varphi n + n \sqrt{n}}{1 + n\sqrt{n}}{\left(an + \left(\left(1- \frac{1}{n^2} - \frac{1}{\sqrt{n}}\right)^{an} -1 \right) \frac{n^2 - n\sqrt{n} - 1}{1 + n\sqrt{n}}\right)}\\
   &= a n\frac{\varphi n + n \sqrt{n}}{1 + n\sqrt{n}} + \left(\left(1- \frac{1}{n^2} - \frac{1}{\sqrt{n}}\right)^{an} -1 \right) \frac{\varphi n + n \sqrt{n}}{1 + n\sqrt{n}} \; \frac{n^2 - n\sqrt{n} - 1}{1 + n\sqrt{n}}\\
   &= \sqrt{n} \left( a \sqrt{n} \frac{\varphi n + n \sqrt{n}}{1 + n\sqrt{n}} + \left(\left(1- \frac{1}{n^2} - \frac{1}{\sqrt{n}}\right)^{an} -1 \right) \frac{\varphi n + n \sqrt{n}}{1 + n\sqrt{n}} \; \frac{n^2 - n\sqrt{n} - 1}{\sqrt{n} + n^2}\right)\;.   
\end{align*}
Thus, the condition $\E[\textsc{Alg}_k] \leq k$ can be written as
\begin{align*}
& \sqrt{n} \left( a \sqrt{n} \frac{\varphi n + n \sqrt{n}}{1 + n\sqrt{n}} + \left(\left(1- \frac{1}{n^2} - \frac{1}{\sqrt{n}}\right)^{an} -1 \right) \frac{\varphi n + n \sqrt{n}}{1 + n\sqrt{n}} \; \frac{n^2 - n\sqrt{n} - 1}{\sqrt{n} + n^2}\right) \leq an\\
\Leftrightarrow & a \sqrt{n} \frac{\varphi n + n \sqrt{n}}{1 + n\sqrt{n}} + \left(\left(1- \frac{1}{n^2} - \frac{1}{\sqrt{n}}\right)^{an} -1 \right) \frac{\varphi n + n \sqrt{n}}{1 + n\sqrt{n}} \; \frac{n^2 - n\sqrt{n} - 1}{\sqrt{n} + n^2} \leq a\sqrt{n}\\
\Leftrightarrow & \left(\left(1- \frac{1}{n^2} - \frac{1}{\sqrt{n}}\right)^{an} -1 \right) \frac{\varphi n + n \sqrt{n}}{1 + n\sqrt{n}} \; \frac{n^2 - n\sqrt{n} - 1}{\sqrt{n} + n^2} \leq a\sqrt{n} \left(1-\frac{\varphi n + n \sqrt{n}}{1 + n\sqrt{n}} \right)\\
\Leftrightarrow & \left(\left(1- \frac{1}{n^2} - \frac{1}{\sqrt{n}}\right)^{an} -1 \right) \frac{\varphi n + n \sqrt{n}}{1 + n\sqrt{n}} \; \frac{n^2 - n\sqrt{n} - 1}{\sqrt{n} + n^2} \leq a \left(\frac{\sqrt{n}-\varphi n\sqrt{n}}{1 + n\sqrt{n}} \right)\;.
\end{align*}
For $n \to \infty$, we get
\begin{align*}
 (0 -1) \cdot 1 \cdot 1 \leq -a \varphi \Leftrightarrow a \leq \frac{1}{\varphi}\;.   
\end{align*}
Thus, for each $\epsilon>0$, there is an $n'$ such that for all $n \geq n'$ the condition is fulfilled for all $k$ with $k \leq (\frac{1}{\varphi} - \epsilon) n$ and not fulfilled for all $k$ with $k \geq (\frac{1}{\varphi} + \epsilon) n$. Thus, $k^\prime \in [(\frac{1}{\varphi} - \epsilon)n, (\frac{1}{\varphi} + \epsilon)n]$.

Now we are ready to derive a formula for the expected value of $\textsc{Alg}$ as follows.
\begin{restatable}{lemma}{loweronl}
\label{lem:lower_onl}
Let $\E[\textsc{Alg}_n]$ denote the expected value of an optimal online algorithm for the instance with $n$ steps and cumulative distribution function $F$ defined above. Let $k^\prime$ be defined as described above. We have
    \begin{align*}\E[\textsc{Alg}_n] &= \sum_{i=1}^{n-k^\prime-1}\left(\left(1-\frac{1}{n^2}\right)^{i-1} \left(\frac{1}{n^2} n \varphi (n-i+1) + \frac{1}{\sqrt{n}}\right)\right) + \left(1-\frac{1}{n^2}\right)^{n-k^\prime-1} \E[\textsc{Alg}_{k^\prime}]\;.
    \end{align*}
Furthermore, it holds for arbitrarily small $\varepsilon > 0$ that
\[\revised{\lim_{n\to \infty} \frac{1}{n} \E[\textsc{Alg}_n] \leq  \frac{1}{\varphi} + \frac{\varphi - 1/\varphi}{2} + \epsilon \left(2-\epsilon \frac{\varphi}{2}\right)\;.}\]
\end{restatable}

\begin{proof}
In the proof of Lemma \ref{lem:lower_onl}, we will need the values of some limits, that we prove afterwards in Lemma \ref{lem:technicalLimits} for a better readability.

For the first statement, we rewrite
\[\E[\textsc{Alg}_n] = \E\left[\sum_{i=1}^n X_i t_i\right] = \sum_{i=1}^n \E[X_i t_i] = \sum_{i=1}^{n-k^\prime-1} \E[X_i t_i] + \sum_{i=n-k^\prime}^{n} \E[X_i t_i] \,. \]

For the first of these two sums, note that as long as the number of remaining steps is larger than $k^\prime$, we have already observed that an optimal online algorithm will select a random variable with value $\varphi n$ for all steps and select all other random variables only for a single step. Thus, for some step $i \in \{1, \dots, n-k^\prime-1\}$, there is a probability of $(1-1/n^2)^{i-1}$ that none of the previous random variables was selected for the whole period. In this case, we see a random variable with value $\varphi n$ with probability $1/n^2$ and select it for $n-i+1$ steps. With probability $1/\sqrt{n}$, we see a random variable with value $1$ and select it for one step.

Finally, the second sum equals zero, in case the algorithm finds a realization with value $\varphi n$ in the first $n-k^\prime-1$ steps. Otherwise, with probability $\left(1-1/n^2\right)^{n-k^\prime-1}$, we get the expected revenue of the algorithm with $k^\prime$ remaining steps, $\E[\textsc{Alg}_{k^\prime}]$, which has been characterized in Lemma \ref{lem:lowersmallk} already.

For the second statement of the lemma, we set $a = 1- \frac{1}{\varphi} - \epsilon$ for some small $\varepsilon > 0$, i.e., for $n$ large enough we have $k^\prime \leq (1-a)n$. We get
\begin{align*}
\lim_{n \to \infty} &\frac{1}{n} \left(1-\frac{1}{n^2}\right)^{n-k^\prime-1}  \E[\textsc{Alg}_{k^\prime}] = \lim_{n \to \infty} \frac{1}{n} \left(1-\frac{1}{n^2}\right)^{n-k^\prime}  \E[\textsc{Alg}_{k^\prime}]\\
&\leq \lim_{n \to \infty} \frac{1}{n} \left(1-\frac{1}{n^2}\right)^{an}  \E[\textsc{Alg}_{(1-a)n}]\\
&= \lim_{n \to \infty}\frac{1}{n} \left(1-\frac{1}{n^2}\right)^{an} \frac{\varphi n + n \sqrt{n}}{1 + n\sqrt{n}}{\left((1-a)n + \left(\left(1- \frac{1}{n^2} - \frac{1}{\sqrt{n}}\right)^{(1-a)n} -1 \right) \frac{n^2 - n\sqrt{n} - 1}{1 + n\sqrt{n}}\right)}\\
&= \lim_{n \to \infty} (1-a) \underbrace{\left(1-\frac{1}{n^2}\right)^{an}}_{\to 1 \text{ by Lemma \ref{lem:technicalLimits} (i)}} \underbrace{\frac{\varphi n + n \sqrt{n}}{1 + n\sqrt{n}}}_{\to 1}\\
&\qquad + \underbrace{\frac{1}{n} \left(1-\frac{1}{n^2}\right)^{an}  \frac{\left(\varphi n + n \sqrt{n}\right)\left(n^2 - n\sqrt{n} - 1\right)}{\left(1 + n\sqrt{n}\right)\left(1 + n\sqrt{n}\right)}}_{\to 0, \text{ as in } \Theta\left(\frac{\sqrt{n}}{n}\right)}\underbrace{\left(\left(1- \frac{1}{n^2} - \frac{1}{\sqrt{n}}\right)^{(1-a)n} -1 \right)}_{\text{in } \mathcal{O}(1), \text{ as } 1-\frac{1}{n^2} - \frac{1}{\sqrt{n}} \in [0,1] \text{ for } n \text{ large enough}}\\
&= 1-a = \frac{1}{\varphi} + \epsilon\;.
\end{align*}
Additionally, for $b=(1-\frac{1}{\varphi} + \epsilon)$ we have $k^\prime + 1 \geq (1-b)n$ for $n$ large enough. Thus,
\begin{align*}
\lim_{n \to \infty} &\frac{1}{n} \sum_{i=1}^{n-k^\prime-1}\left(\left(1-\frac{1}{n^2}\right)^{i-1} \left(\frac{1}{n^2} n \varphi (n-i+1) + \frac{1}{\sqrt{n}}\right)\right)\\
&\leq \lim_{n \to \infty} \frac{1}{n} \sum_{i=1}^{n-(1-b)n}\left(\left(1-\frac{1}{n^2}\right)^{i-1} \left(\frac{1}{n^2} n \varphi (n-i+1) + \frac{1}{\sqrt{n}}\right)\right)\\
&= \lim_{n \to \infty} \left(\frac{1}{n} \sum_{i=1}^{bn}\left(1-\frac{1}{n^2}\right)^{i-1} \left(\frac{1}{\sqrt{n}} + \varphi \frac{n+1}{n}\right) - \varphi \frac{1}{n^2} \sum_{i=1}^{bn} \left(\left(1-\frac{1}{n^2}\right)^{i-1} i\right)\right)\\
&= \lim_{n \to \infty} \left(\left(\frac{1}{n\sqrt{n}} + \varphi \frac{n+1}{n^2}\right) \left(\sum_{i=1}^{bn} \left(1-\frac{1}{n^2}\right)^{i-1}\right) - \varphi \frac{1}{n^2} \sum_{i=1}^{bn} \left(\left(1-\frac{1}{n^2}\right)^{i-1} i\right)\right)\\
&= \lim_{n \to \infty} \left(\left(\frac{1}{n\sqrt{n}} + \varphi \frac{n+1}{n^2}\right) \left(n^2\left( 1-\left(1-\frac{1}{n^2}\right)^{bn}\right)\right) - \varphi \left(n^2 - \left(1-\frac{1}{n^2}\right)^{bn} \left(n^2+nb \right)\right)\right)\\
&= \lim_{n \to \infty} \Bigg(\left(\frac{1}{\sqrt{n}} + \varphi \frac{n+1}{n}\right) \left(n\left(1-\left(1-\frac{1}{n^2}\right)^{bn}\right)\right)\\
& \quad + \varphi n \left( n \left(\left(1-\frac{1}{n^2}\right)^{bn} -1 \right) +b \right) + \varphi b n \left(\left(1-\frac{1}{n^2}\right)^{bn} -1\right) + n^2 \varphi - \varphi b n + \varphi b n - n^2 \varphi \Bigg)\\
&= \lim_{n \to \infty} \Bigg(\underbrace{\left(\frac{1}{\sqrt{n}} + \varphi \frac{n+1}{n}\right)}_{\to \varphi} \underbrace{\left(n\left(1-\left(1-\frac{1}{n^2}\right)^{bn}\right)\right)}_{\to b \text{ by Lemma \ref{lem:technicalLimits} (ii)}}\\
& \quad + \underbrace{\varphi n \left( n \left(\left(1-\frac{1}{n^2}\right)^{bn} -1 \right) +b \right)}_{\to \varphi \frac{b^2}{2} \text{ by Lemma \ref{lem:technicalLimits} (iii)}} + \underbrace{\varphi b n \left(\left(1-\frac{1}{n^2}\right)^{bn} -1\right)}_{\to \varphi b (-b) \text{ by Lemma \ref{lem:technicalLimits} (ii)}}\Bigg)\\
&= \varphi b + \varphi \frac{b^2}{2} - \varphi b^2 = \varphi b \left(1-\frac{b}{2}\right) = \varphi \left(1-\frac{1}{\varphi} + \epsilon\right) \left(\frac{1+ \frac{1}{\varphi} - \epsilon}{2}\right) = \frac{\varphi}{2} \left(1-\left(\frac{1}{\varphi} - \epsilon\right)^2\right)\\
&= \frac{\varphi}{2} - \frac{\varphi}{2} \left(\frac{1}{\varphi^2} - 2 \epsilon \frac{1}{\varphi} + \epsilon^2\right) = \frac{\varphi}{2} - \frac{1}{2\varphi} + \epsilon \left(1-\epsilon \frac{\varphi}{2}\right) = \frac{\varphi - 1/\varphi}{2} + \epsilon \left(1-\epsilon \frac{\varphi}{2}\right)\;.
\end{align*}

Together with the other part of the limit we obtain
\[\revised{\lim_{n\to \infty} \frac{1}{n} \E[\textsc{Alg}_n] \leq \left(\frac{1}{\varphi} + \epsilon\right) + \left(\frac{\varphi - 1/\varphi}{2} + \epsilon \left(1-\epsilon \frac{\varphi}{2}\right)\right) = \frac{1}{\varphi} + \frac{\varphi - 1/\varphi}{2} + \epsilon \left(2-\epsilon \frac{\varphi}{2}\right)\;,}\]
which finishes the proof.
\end{proof}

Next, we deliver a proof of the technical results referenced in the preceding proof.

\begin{restatable}{lemma}{technicalLimits}
For constant $a \geq 0$, the following three statements are true:
\begin{enumerate}
\item[(i)] $\lim_{n \to \infty} \left(1-\frac{1}{n^2}\right)^{an} = 1$,
\item[(ii)] $\lim_{n \to \infty} n \left(\left(1-\frac{1}{n^2}\right)^{an} -1 \right) = -a$,
\item[(iii)] $\lim_{n \to \infty} n \left( n \left(\left(1-\frac{1}{n^2}\right)^{a n} -1 \right) + a\right) = \frac{a^2}{2}$.
\end{enumerate}
\label{lem:technicalLimits}
\end{restatable}

\begin{proof}
In order to prove $(i)$, observe that
\begin{align*}
 \lim_{n\to \infty} \left(1-\frac{1}{n^2}\right)^{a n} &= \lim_{n\to \infty} \exp \left( \ln \left(\left(1-\frac{1}{n^2}\right)^{a n}\right)\right) =  \lim_{n\to \infty} \exp \left(a n \ln \left(\left(1-\frac{1}{n^2}\right)\right)\right)\\
 &= \exp \left( \lim_{n\to \infty} a n \ln \left(\left(1-\frac{1}{n^2}\right)\right)\right) = \exp \left( a \lim_{n\to \infty} \frac{\ln \left(\left(1-\frac{1}{n^2}\right)\right)}{\frac{1}{n}}\right)\\
 &\overset{\text{l'H\^opital}}{=} \exp \left( a \lim_{n\to \infty} \frac{\frac{2}{n^3-n}}{-\frac{1}{n^2}}\right) = \exp \left( a \lim_{n\to \infty} \frac{-2 n^2}{n^3-n}\right) = 1\;.
 \end{align*}
 
Before showing $(ii)$, we calculate
\[\frac{\partial}{\partial n} \left(1-\frac{1}{n^2}\right)^{a n} = \frac{a(1 - \frac{1}{n^2})^{a n} (2 + (n^2-1) \ln(1 - \frac{1}{n^2}))}{(n^2-1)}\;,\]
and
\begin{equation}\label{eq:lnlimit}
    \lim_{n \to \infty} {(n^2-1) \ln \left(1-\frac{1}{n^2}\right)} =  \ln \lim_{n \to \infty} {\left(1-\frac{1}{n^2}\right)^{n^2-1}} = \ln \frac{1}{e} = -1\;.
\end{equation}
Together with $(i)$, we get
\begin{align*}
 \lim_{n\to \infty} n\left(\left(1-\frac{1}{n^2}\right)^{a n} - 1\right) &= \lim_{n\to \infty} \frac{\left(\left(1-\frac{1}{n^2}\right)^{a n} - 1\right)}{\frac{1}{n}} \overset{\text{l'H\^opital}}{=} \lim_{n\to \infty} \frac{\frac{a(1 - \frac{1}{n^2})^{a n} (2 + (n^2-1) \ln(1 - \frac{1}{n^2}))}{(n^2-1)}}{-\frac{1}{n^2}}\\
 &=  \lim_{n\to \infty} -\frac{n^2}{n^2-1} a\left(1 - \frac{1}{n^2}\right)^{a n} \left(2 + (n^2-1) \ln\left(1 - \frac{1}{n^2}\right)\right)\\
 &= -1(a \cdot 1 (2+(-1))) = -a\;,
 \end{align*}
which finishes the proof of $(ii)$.

In order to show $(iii)$, we need the following limit.

\begin{align*}
\lim_{n\to \infty}\left( 2n + 2n^3 \ln\left(1-\frac{1}{n^2}\right)\right) &= 2 \lim_{n\to \infty} \frac{1 + n^2 \ln\left(1-\frac{1}{n^2}\right)}{\frac{1}{n}}
\end{align*}
Analogously to Equation \ref{eq:lnlimit}, $\lim_{n \to \infty}n^2 \ln\left(1-\frac{1}{n^2}\right) = -1$, so we get by l'H\^opital's rule
\begin{align*}
&= 2 \lim_{n\to \infty} \frac{2n\left(\frac{1}{n^2-1} + \ln\left(1-\frac{1}{n^2}\right)\right)}{-\frac{1}{n^2}} = 2 \lim_{n\to \infty}\left( -\frac{2n^3}{n^2-1} - 2n^3\ln\left(1-\frac{1}{n^2}\right)\right)\\
&= -4 \lim_{n\to \infty}\left( \frac{n^3}{n^2-1} + n^3\ln\left(1-\frac{1}{n^2}\right)\right) = -4 \lim_{n\to \infty}\left( n + n^3\ln\left(1-\frac{1}{n^2}\right) + \frac{n}{n^2-1}\right)\;.
\end{align*}
Suppose, $\lim_{n\to \infty} n + n^3 \ln\left(1-\frac{1}{n^2}\right)=c$ for some $c \in \mathbb{R} \cup \{-\infty, \infty\}$. Then, by limit rules we get
\[2c = -4 (c + 0)\;,\]
i.e., we conclude
\begin{equation}\label{eq:helplimit1}
    \lim_{n\to \infty} 2n + 2n^3 \ln\left(1-\frac{1}{n^2}\right) = 0\;.
\end{equation}
Additionally, we will first show that
\begin{equation}\label{eq:helplimit2}
    \lim_{n\to \infty}n\left(1 + \left(1-\frac{1}{n^2}\right)^{a n} (n^2-1) \ln\left(1-\frac{1}{n^2}\right)\right) = a\;.
\end{equation}
In fact,
\begin{align*}
\lim_{n\to \infty}&n\left(1 + \left(1-\frac{1}{n^2}\right)^{a n} (n^2-1) \ln\left(1-\frac{1}{n^2}\right)\right) = \lim_{n\to \infty}\left(\frac{1 + \left(1-\frac{1}{n^2}\right)^{a n} (n^2-1) \ln\left(1-\frac{1}{n^2}\right)}{\frac{1}{n}}\right)\;.
\end{align*}
Since $\lim_{n \to \infty} \left(1-\frac{1}{n^2}\right)^{a n} = 1$ by (i) and $\lim_{n \to \infty} {(n^2-1) \ln \left(1-\frac{1}{n^2}\right)} = -1$ due to Equation \eqref{eq:lnlimit}, we can apply l'H\^opital's rule and obtain
\begin{align*}
&= \lim_{n\to \infty} \frac{\frac{1}{n}\left(1-\frac{1}{n^2}\right)^{a n} \left(2+2n(a+n)\ln\left(1-\frac{1}{n^2}\right) + a n \left(n^2-1\right)\left(\ln\left(1-\frac{1}{n^2}\right)\right)^2\right)}{-\frac{1}{n^2}}\\
&= \lim_{n\to \infty} - \left(1-\frac{1}{n^2}\right)^{a n}\left(2n+2n^2(a+n)\ln\left(1-\frac{1}{n^2}\right) + a n^2 \ln \left(1-\frac{1}{n^2}\right) (n^2-1) \ln \left(1-\frac{1}{n^2}\right)\right)\\
&= \lim_{n\to \infty} \underbrace{- \left(1-\frac{1}{n^2}\right)^{a n}}_{\to -1 \text{ by (i)}}\\
&\qquad \cdot \left(\underbrace{2n+2n^3\ln\left(1+\frac{1}{n^2}\right)}_{\to 0 \text{ by Equation \eqref{eq:helplimit1}}} + \underbrace{2a n^2 \ln\left(1-\frac{1}{n^2}\right)}_{\to -2a} + \underbrace{a n^2 \ln \left(1-\frac{1}{n^2}\right)}_{\to -a} \underbrace{(n^2-1) \ln \left(1-\frac{1}{n^2}\right)}_{\to -1}\right)\\
&=a\;.
\end{align*}
We proceed by showing $(iii)$.
\begin{align*}
  &\lim_{n \to \infty} n \left( n \left(\left(1-\frac{1}{n^2}\right)^{a n} -1 \right) + a\right) = \lim_{n\to \infty} n^2 \left(\left(1-\frac{1}{n^2}\right)^{a n} - 1 + \frac{a}{n}\right)\\
  &= \lim_{n\to \infty} \frac{\left(1-\frac{1}{n^2}\right)^{a n} - 1 + \frac{a}{n}}{\frac{1}{n^2}} \overset{\text{l'H\^opital}}{=} \lim_{n\to \infty} \frac{\frac{\left(1 - \frac{1}{n^2}\right)^{a n} \left(2a + \left(n^2-1\right) a\ln\left(1 - \frac{1}{n^2}\right)\right)}{(n^2-1)} - \frac{a}{n^2}}{-\frac{2}{n^3}}\\
  &=  \frac{1}{2} \lim_{n\to \infty} -\frac{n^3}{n^2-1} \left(1 - \frac{1}{n^2}\right)^{a n} \left(2a + \left(n^2-1\right) a\ln\left(1 - \frac{1}{n^2}\right)\right) + \frac{a n^3}{n^2}\\
  &= \frac{1}{2} \lim_{n \to \infty} \frac{n^2}{n^2-1} \left(-2 a n\left(1-\frac{1}{n^2}\right)^{a n} - a n\left(1-\frac{1}{n^2}\right)^{a n} \left(n^2-1 \right) \ln \left(1-\frac{1}{n^2}\right) + a\left(n-\frac{1}{n}\right)\right)\\
  &= \frac{1}{2} \lim_{n\to \infty} \frac{n^2}{n^2-1} \left(-\frac{a}{n} - 2 a n\left(\left(1-\frac{1}{n^2}\right)^{a n} -1\right) - a n -a n\left(1-\frac{1}{n^2}\right)^{a n} \left(n^2-1\right) \ln \left(1-\frac{1}{n^2}\right)\right)\\
  &= \frac{1}{2} \lim_{n\to \infty} \underbrace{\frac{n^2}{n^2-1}}_{\to 1} \left(\underbrace{-\frac{a}{n}}_{\to 0} - \underbrace{2 a n\left(\left(1-\frac{1}{n^2}\right)^{an} -1\right)}_{\to -2a^2 \text{ by $(ii)$}} - \underbrace{an\left(1 + \left(1-\frac{1}{n^2}\right)^{a n} (n^2-1) \ln\left(1-\frac{1}{n^2}\right)
  \right)}_{\to a^2 \text{ by Equation \eqref{eq:helplimit2}}}\right)\\
  &=\frac{a^2}{2}\;
  \end{align*}
  which finishes the proof of (3) and thus Lemma \ref{lem:technicalLimits}.
\end{proof}

Next, we turn our attention to the performance of the optimal offline procedure. The following lemma provides an exact formula for the expected profit for the omniscient prophet for the given example instance.
\begin{restatable}{lemma}{loweroff}
\label{lem:lower_off}
Let $\E[\textsc{Opt}_n]$ denote the expected profit of the offline algorithm for the instance with $n$ steps and cumulative distribution function $F$ defined above. We have
    \[\E[\textsc{Opt}_n] = \sum_{i=1}^{n}\left(\varphi n \left(1-\left(1-\frac{1}{n^2}\right)^i\right) + \left(1-\frac{1}{n^2}\right)^i - \left(1-\frac{1}{\sqrt{n}} - \frac{1}{n^2}\right)^i \right)\;. \]
Furthermore, it holds that
\[\lim_{n\to \infty}{\frac{1}{n} \E[\textsc{Opt}_n]} = \frac{\varphi}{2} + 1\;.\]
\end{restatable}

\begin{proof}
Note that in all steps the selected random variable is the one that has the best value seen so far. Thus, we get $\varphi n$ in step $i$ unless we have not seen a random variable with value $\varphi n$ in all steps $1$ to $i$, which happens with probability $(1-1/n^2)^i$. Similarly, with probability $(1-1/n^2)^i - (1-1/\sqrt{n} - 1/n^2)^i$ we have not seen a random variable with value $\varphi n$ and not seen any 0-items, which implies that the largest value seen so far is equal to 1.

It remains to show the second statement. In order to do so, we first observe
\begin{align*}
\lim_{n \to \infty} {\frac{1}{n}\sum_{i=1}^n \left(1-\frac{1}{\sqrt{n}} - \frac{1}{n^2}\right)^i }&= \lim_{n \to \infty} \frac{1}{n} \frac{\left(1-\left(1-\frac{1}{\sqrt{n}}-\frac{1}{n^2}\right)^n\right) \left(n^2-n^{3/2}-1\right)}{n^{3/2}+1}\\
& \leq \lim_{n \to \infty} \frac{1}{n} \frac{n^2-n^{3/2}-1}{n^{3/2}+1} = 0\;.
\end{align*}
Note that $1-\frac{1}{\sqrt{n}} - \frac{1}{n^2}$ is positive for $n \geq 2$, and thus $\lim_{n \to \infty} {\frac{1}{n}\sum_{i=1}^n \left(1-\frac{1}{\sqrt{n}} - \frac{1}{n^2}\right)^i } \geq 0$. Next, we observe that $\lim_{n \to \infty} {\frac{1}{n}\sum_{i=1}^n \left(1-\frac{1}{\sqrt{n}} - \frac{1}{n^2}\right)^i } =0$. Now, we are ready to prove the second statement of the lemma. By the following technical reformulation we obtain
\begin{align*}
\lim_{n \to \infty} \frac{1}{n} \E[\textsc{Opt}_n] &= \sum_{i=1}^{n}\left(\varphi n \left(1-\left(1-\frac{1}{n^2}\right)^i\right) + \left(1-\frac{1}{n^2}\right)^i - \left(1-\frac{1}{\sqrt{n}} - \frac{1}{n^2}\right)^i \right)\\
&=\lim_{n \to \infty} \frac{1}{n}\sum_{i=1}^{n}\left(\varphi n \left(1-\left(1-\frac{1}{n^2}\right)^i\right) + \left(1-\frac{1}{n^2}\right)^i - \left(1-\frac{1}{\sqrt{n}} - \frac{1}{n^2}\right)^i \right)\\
&=\lim_{n \to \infty} \varphi n - \varphi \sum_{i=1}^{n}\left(1-\frac{1}{n^2}\right)^i + \frac{1}{n} \sum_{i=1}^n\left(1-\frac{1}{n^2}\right)^i - \frac{1}{n}\sum_{i=1}^n \left(1-\frac{1}{\sqrt{n}} - \frac{1}{n^2}\right)^i \\
&=\lim_{n \to \infty} \varphi n + \left(\frac{1}{n} - \varphi\right)  \sum_{i=1}^{n}\left(1-\frac{1}{n^2}\right)^i - \frac{1}{n}\sum_{i=1}^n \left(1-\frac{1}{\sqrt{n}} - \frac{1}{n^2}\right)^i \\
&=\lim_{n \to \infty} \varphi n + \left(\frac{1}{n} - \varphi\right) \left(1 - \left(1-\frac{1}{n^2}\right)^n\right) \left(n^2-1\right)\\
&=\lim_{n \to \infty} \varphi n + \left(1-\left(1-\frac{1}{n^2}\right)^n\right) \left(n - \frac{1}{n} - n^2\varphi + \varphi\right)\\
&=\lim_{n \to \infty} \varphi n + n - \frac{1}{n} - n^2 \varphi + \varphi + \left(1-\frac{1}{n^2}\right)^n \left(n^2 \varphi - n - \varphi + \frac{1}{n}\right)\\
&= \lim_{n \to \infty} \underbrace{n\varphi \left(n \left(\left(1-\frac{1}{n^2}\right)^n - 1\right) + 1\right)}_{\to \frac{\varphi}{2} \text{by Lemma \ref{lem:technicalLimits} (iii) for }a=1} \underbrace{- n\left(\left(1-\frac{1}{n^2}\right)^n -1 \right)}_{\to 1 \text{ by Lemma \ref{lem:technicalLimits} (ii) for }a=1}\\
&\qquad \underbrace{- \varphi \left( \left( 1-\frac{1}{n^2}\right)^n - 1\right)}_{\to 0} + \underbrace{\frac{1}{n} \left(\left(1-\frac{1}{n^2}\right)^n -1 \right)}_{\to 0}\\
&= \frac{\varphi}{2} + 1\;,
\end{align*}
finishing the proof of the lemma.
\end{proof}

We combine the two lemmas to derive the desired result.

\begin{proof}[Proof of Theorem \ref{thm:lower}]
By Lemma \ref{lem:lower_onl} and Lemma \ref{lem:lower_off}, it holds for arbitrarily small $\varepsilon > 0$ that
\begin{align*}
\revised{\lim_{n \to \infty} \frac{\E[\textsc{Alg}_n]}{\E[\textsc{Opt}_n]} \leq \frac{\frac{1}{\varphi} + \frac{\varphi - 1/\varphi}{2} + \epsilon \left(2-\epsilon \frac{\varphi}{2}\right)}{\frac{\varphi}{2} + 1}\;.}
\end{align*}
For $\epsilon \to 0$ this ratio is arbitrarily close to
\begin{align*}
\revised{\frac{\frac{1}{\varphi} + \frac{\varphi - 1/\varphi}{2}}{\frac{\varphi}{2} + 1} \;,}
\end{align*}
\revised{which, assuming the parameter $\varphi$ is positive, is minimized for the choice of $\varphi$ as the golden ratio. Due to its property $\varphi - \frac{1}{\varphi} = 1$, we finally get the bound}
\begin{align*}
\revised{\lim_{n \to \infty} \frac{\E[\textsc{Alg}_n]}{\E[\textsc{Opt}_n]} \leq \frac{\frac{1}{\varphi} + \frac{1}{2}}{\frac{\varphi}{2} + 1} = \frac{2\left(2 + \varphi \right)}{2 \varphi \left(\varphi + 2\right)} = \frac{1}{\varphi}\;.}
\end{align*}
\end{proof}

In the classical (non-over time) setting, a single-threshold procedure suffices to achieve the optimal prophet inequality~\cite{samuel-cahn83}. Interestingly, in our setting, we can separate simple single-threshold strategies from our multi-threshold online procedure shown in Section \ref{sec:betterlower}.\vspace{6pt plus 3pt minus 2pt}

\begin{remark}
\revised{\textit{Any deterministic online procedure that bases its decision in the current step only on the drawn value and not on the number of remaining steps achieves a strictly worse prophet inequality than $0.570$ for $n\to \infty$. It will turn out that this is strictly worse than the prophet inequality of the procedure shown in Section \ref{sec:betterlower}.}}\vspace{6pt plus 3pt minus 2pt}
\end{remark}

To see this, consider a best-possible online procedure that bases all decisions only on the drawn value $x$ from the distribution $F^\prime$ given by the random variable
\[X = \begin{cases}2 n & \text{with prob.\ }\frac{1}{n^2}\;,\\
1 & \text{with prob.\ }\frac{1}{\sqrt{n}}\;,\\
0 & \text{with prob.\ }1 - \frac{1}{\sqrt{n}} - \frac{1}{n^2}\;.\end{cases}\]

Note that $F^\prime$ differs from the distribution $F$ used for Theorem \ref{thm:lower} only in the highest possible realization, which is now $2n$ instead of $\varphi n$. In a best-possible online procedure, a random variable with value $2 n$ will be selected for all remaining steps, and a random variable with value $0$ will never be selected. Consequently, the only freedom for optimization is in the duration for which a random variable with value 1 is selected.

By arguments similar to the ones in Section \ref{structureOnline}, we observe that this is either one single step or all remaining steps. However, since we only allow continuous distributions in our model, a continuous version of $F'$ is presented to the algorithm, as discussed at the end of Section \ref{sec:Model}. This additional randomness is equivalent to the online procedure being enabled to perform a $p$-coinflip on each arrival of a realization $x$. Notably, the probability $p$ cannot depend on the number of remaining steps, but it could depend on the number of total steps $n$.

It is straightforward to adapt the formulas derived in this section to calculate the prophet inequalities for the instance $F^\prime$ with a given parameter $p\in[0,1]$. It turns out that the prophet inequality is maximized for the choice of ${p\approx 2.688/\sqrt{n}}$ and its maximum is upper bounded by $0.570$. Together with the lower bound of $0.598$ for the multi-threshold algorithm shown in the subsequent section for all instances, we can separate single- from multi-threshold algorithms in our setting.

%%%%%%%%%%%
% SECTION %
%%%%%%%%%%%

\section{A Better Lower Bound on the Prophet Inequality over Time}\label{sec:betterlower}

In this section, we introduce a more sophisticated algorithm, which has the exact same structure as the optimal Algorithm $\ref{alg:onlineopt}$, but uses different thresholds that permit analysis.

\begin{algorithmEnv}[ht]
\begin{algorithm}[H]\caption{\textsc{Onl}}\label{alg:onl}
\For{$i=1,\dots,n$}{
\eIf{$x_i > \delta_{p(i)}$}{
select random variable for all $n-i+1$ remaining steps;\\
\textbf{break}
}{select random variable for one step;}}
 \end{algorithm}
  \caption{An Advanced Decision Procedure for Prophet Inequalities over Time}
\end{algorithmEnv}

The algorithm's threshold in step $i=1,\dots,n$ is the quantile $\delta_{p(i)}$, where $p(i)$ denotes some probability. For technical reasons, let $p(0)=1$, $p(n+1)=p(n+2)=0$. A large part of our analysis is independent of $p(i)$ and thus can be used as a building block for other decision rules.

\begin{restatable}{theorem}{onlbound}\label{thm:onlbound}
   The algorithm \textsc{Onl} with threshold $\delta_{p(i)}$ where $p(i) = \e^{-\frac{c i}{n^2}}$ and $c = 9.71$ gives a lower bound of $0.598$ for the prophet inequality over time for $n\rightarrow\infty$.
\end{restatable}

For the proof, we separately bound the expected value of the algorithm $\E[\textsc{Onl}_n]$ and the offline optimal solution $\E[\textsc{Opt}_n]$. We express both expectations as a weighted sum over the expected value that the algorithm extracts from realized random variables in the range of the distribution between two adjacent quantiles  $\delta_{p(k+1)} < x \leq \delta_{p(k)}$. This results in the following two lemmas for monotonically decreasing probabilities $p(k)$. As the offline optimal algorithm selects the maximum value seen so far in each step, we bound this value using a fragmentation into possible events.

\begin{restatable}{lemma}{alphaoff}\label{lem:alphaoff}
Let $n\in\mathbb{N}_{>0}$. For quantiles $\delta_{p(k)}$, given a monotonically decreasing function $p(k)$ for $k\in[n]$ with $p(0)=1$ and $p(n+1)=p(n+2)=0$, it holds that
    \[\E[\textsc{Opt}_n] \leq \sum_{k=0}^n \alpha_k^* \E[X \mid \delta_{p(k+1)} < X \leq \delta_{p(k)}]\;,\]
    where \[\alpha_0^*=(1-p(1))\frac{n(n+1)}{2} + \sum_{i=2}^n \left(p(1)^i-p(2)^i-i \cdot p(2)^{i-1} (p(1)-p(2))\right),\] and for all $k\in\{1, \dots, n\}$,\\
    \begin{align*}
        \alpha_k^*&=\sum_{i=1}^n i \cdot p(k+1)^{i-1} (p(k)-p(k+1))\\ &\qquad+ \sum_{i=2}^n \left(p(k+1)^i - p(k+2)^i - i \cdot p(k+2)^{i-1} (p(k+1)-p(k+2))\right).
    \end{align*}
\end{restatable}

\begin{proof}
    Recall that \[\E[\textsc{Opt}_n] = \E\left[\sum_{i=1}^n \max\{X_1,\dots,X_i\}\right] = \sum_{i=1}^n \E[\max\{X_1,\dots,X_i\}].\]
    Again, we provide an upper bound for each summand separately. In order to do so, for all $i\in[n]$ and $k\in[n]$, let
    \[\mathcal{E}_{i,k,=1}=[|\{x_1,\dots,x_i\} \cap (\delta_{p(k+1)}, \delta_{p(k)}]| = 1 \text{ and } |\{x_1,\dots,x_i\} \cap (\delta_{0}, \delta_{p(k)}]| = i|]\]
    be the event where the maximum of $i$ draws is in the interval $(\delta_{p(k+1)}, \delta_{p(k)}]$ but all other $i-1$ draws are at most $\delta_{p(k+1)}$. Analogously, we define
    \[\mathcal{E}_{i,k,>1}=[|\{x_1,\dots,x_i\} \cap (\delta_{p(k+1)}, \delta_{p(k)}]| > 1 \text{ and } |\{x_1,\dots,x_i\} \cap (\delta_{0}, \delta_{p(k)}]| = i|]\]
    to be the event where the maximum of $i$ draws is again in the interval $(\delta_{p(k+1)}, \delta_{p(k)}]$, but at least one more of the other $i-1$ draws is in the same interval.
    Lastly, for $k=0$, we define
    \[\mathcal{E}_{i,0,\geq1}=[|\{x_1,\dots,x_i\} \cap (\delta_{p(1)}, \delta_{p(0)}]| \geq 1]\]
    as the event that the maximum of $i$ draws is larger than the largest threshold $\delta_{p(1)}$. Since for fixed $i\in[n]$,
    \[\Prob[\mathcal{E}_{i,0,\geq1}] + \sum_{k=1}^{n} (\Prob[\mathcal{E}_{i,k,=1}] + \Prob[\mathcal{E}_{i,k,>1}]) = 1,\]
    we can use the law of total probability in order to derive that
    \begin{align}\label{eq:lot}
        \begin{split}
        \E[\max\{X_1,\dots,X_i\}] = & \Prob[\mathcal{E}_{i,0,\geq1}] \cdot \E[\max\{X_1,\dots,X_i\} \mid \mathcal{E}_{i,0,\geq1}]\\
        &+\sum_{k=1}^{n} (\Prob[\mathcal{E}_{i,k,=1}] \cdot \E[\max\{X_1,\dots,X_i\} \mid \mathcal{E}_{i,k,=1}]\\
        &+ \Prob[\mathcal{E}_{i,k,>1}] \cdot \E[\max\{X_1,\dots,X_i\} \mid \mathcal{E}_{i,k,>1}]).
        \end{split}
    \end{align}
    Next, we determine upper bounds on the occurring conditional expected values. For $k\in[n]$, we have that
    \begin{itemize}
        \item $\E[\max\{X_1,\dots,X_i\} \mid \mathcal{E}_{i,k,=1}]=\E[X \mid \delta_{p(k+1)} < X \leq \delta_{p(k)}]$, and
        \item $\E[\max\{X_1,\dots,X_i\} \mid \mathcal{E}_{i,k,>1}] \leq \delta_{p(k)} < \E[X \mid \delta_{p(k)} < X \leq \delta_{p(k-1)}]$,
    \end{itemize}
    and for $k=0$, we can count the number $j$ of draws in the interval $(\delta_{p(1)},\delta_{p(0)}]$ through analogously defined events $\mathcal{E}_{i,0,=j}$. Using the inequality $\E[\max\{X_1,\dots,X_j\}] \leq j \E[X_1]$ yields
    \begin{itemize}
        \item $
        \begin{aligned}[t]
         \E[\max\{X_1,\dots,X_i\} \mid &\mathcal{E}_{i,0,\geq1}]
            = \sum_{j=1}^i \frac{\Prob[\mathcal{E}_{i,0,=j}]}{\Prob[\mathcal{E}_{i,0,\geq1}]} \E[\max\{X_1,\dots,X_i\} \mid \mathcal{E}_{i,0,=j}]\\
            &\leq \frac{1}{\Prob[\mathcal{E}_{i,0,\geq1}]} \sum_{j=1}^i \binom{i}{j} (1-p(1))^j p(1)^{i-j} j \E[X \mid \delta_{p(1)} < X \leq \delta_{p(0)}]\\
            &= \frac{1}{\Prob[\mathcal{E}_{i,0,\geq1}]} i(1-p(1)) \E[X \mid \delta_{p(1)} < X \leq \delta_{p(0)}],
        \end{aligned}
        $
    \end{itemize}
    where the latter sum equals the expected value of a $(i,1-p(1))$-distributed binomial variable, which is $i(1-p(1))$. Applying the three listed bounds on \autoref{eq:lot} yields that
    \begin{align*}
        \E[\max\{X_1,\dots,X_i\}] \leq &i(1-p(1)) \E[X \mid \delta_{p(1)} < X \leq \delta_{p(0)}]\\
        &+\Prob[\mathcal{E}_{i,1,>1}] \cdot E[X \mid \delta_{p(1)} < X \leq \delta_{p(0)}]\\
        &+\sum_{k=1}^{n-1} (\Prob[\mathcal{E}_{i,k,=1}] + \Prob[\mathcal{E}_{i,k+1,>1}]) \cdot \E[X \mid \delta_{p(k+1)} < X \leq \delta_{p(k)}]\\
        &+\Prob[\mathcal{E}_{i,n,=1}] \cdot E[X \mid \delta_{p(n+1)} < X \leq \delta_{p(n)}].
    \end{align*}
    Next, we determine the occurring probabilities, which are
    \begin{itemize}
        \item $\Prob[\mathcal{E}_{i,k,=1}] =
        \begin{cases}
            p(n), &\text{if } k=n,i=1\\
            0, &\text{if } k=n,i>1\\
            i \cdot (p(k)-p(k+1)) p(k+1)^{i-1}, &\text{if } k\in[n-1]
        \end{cases}$\\
        \item For all $k=0\dots,n-1$ and $i>1$,
        \begin{align*}
            \Prob[\mathcal{E}_{i,k+1,>1}] &= \Prob[\forall j\in[i] : x_j \leq \delta_{p(k+1)}] - \Prob[\forall j\in[i] : x_j \leq \delta_{p(k+2)}] - \Prob[\mathcal{E}_{i,k+1,=1}]\\
            &= p(k+1)^i - p(k+2)^i - i \cdot (p(k+1)-p(k+2)) p(k+2)^{i-1},
        \end{align*}
        \item[$~~$] otherwise $\Prob[\mathcal{E}_{i,k+1,>1}] = 0$.
    \end{itemize}
    Plugging the probabilities into the previous inequality and summing over all steps $i\in[n]$, we overall obtain
    \begin{align*}
        \E[\textsc{Opt}_n] \leq & \left(\sum_{i=1}^n i (1-p(1)) + \sum_{i=2}^n p(1)^i - p(2)^i - i (p(1)-p(2)) p(2)^{i-1} \right) \E[X \mid \delta_{p(1)} < X \leq \delta_{p(0)}]\\
        &+ \sum_{k=1}^{n-1}\bigg(\sum_{i=2}^n p(k+1)^i - p(k+2)^i - i \cdot (p(k+1)-p(k+2)) p(k+2)^{i-1}\\
        &\qquad+ \sum_{i=1}^n i \cdot (p(k)-p(k+1)) p(k+1)^{i-1}\bigg) \E[X \mid \delta_{p(k+1)} < X \leq \delta_{p(k)}]\\
        &+ p(n) \E[X \mid \delta_{p(n+1)} < X \leq \delta_{p(n)}].
    \end{align*}
    The statement of the lemma follows immediately.
\end{proof}

For the analysis of \textsc{Onl}, we bound in each step the probability that the current value falls into each range, that it is selected, and the number of steps for which it is selected. 
\begin{lemma}\label{lem:alphaon}
Let $n\in\mathbb{N}_{>0}$. For quantiles $\delta_{p(k)}$, given a non-increasing function $p(k)$ for $k\in[n]$ with $p(0)=1$ and $p(n+1)=0$, it holds that
    \[\E[\textsc{Onl}_n] = \sum_{k=0}^n \alpha_k \E[X \mid \delta_{p(k+1)} < X \leq \delta_{p(k)}]\;,\]
    where for all $k\in\{0,\dots,n\}$, \[\alpha_k = (p(k)-p(k+1)) \left(\sum_{i=1}^k \prod_{j=1}^{i-1} p(j) + \sum_{i=k+1}^n (n-i+1) \prod_{j=1}^{i-1} p(j)\right).\]
\end{lemma}

\begin{proof}
	We rewrite
	\[	\E[\textsc{Onl}_n] = \E\left[\sum_{i=1}^n X_i t_i\right] = \sum_{i=1}^n \E[X_i t_i] \,.\]
	In case step $i\in[n]$ is the first step in which some realization exceeds the threshold of the corresponding step, which occurs with probability $(\prod_{j=1}^{i-1} p(j))(1-p(i))$, an expected revenue of $\E[X \mid X > \delta_{p(i)}]$ is collected for all remaining $(n-i+1)$ steps. With probability $\prod_{j=1}^{i} p(j)$, no realization in the first $i$ steps exceeded a threshold of the algorithm, leading to an expected value of $\E[X \mid X \leq \delta_{p(i)}]$ being collected for one step. Otherwise, $X_i t_i$ realizes at $0$, as the algorithm is already blocked in step $i$. We can reformulate
	\begin{align*}
	\E[\textsc{Onl}_n]=\sum_{i=1}^n \left(\prod_{j=1}^{i-1} p(j)\right) (1-p(i)) (n-i+1) \E[X \mid X > \delta_{p(i)}] + \sum_{i=1}^n \left(\prod_{j=1}^{i} p(j)\right) \E[X \mid X \leq \delta_{p(i)}] .
	\end{align*}
    Now observe that:
    \begin{itemize}
        \item $
        \begin{aligned}[t]
         \E[X \mid X > \delta_{p(i)}] = \frac{1}{1-p(i)} \sum_{k=0}^{i-1} (p(k)-p(k+1)) \E[X \mid \delta_{p(k+1)} < X \leq \delta_{p(k)}]
        \end{aligned}
        $, and
        \item $
        \begin{aligned}[t]
         \E[X \mid X \leq \delta_{p(i)}] = \frac{1}{p(i)} \sum_{k=i}^{n} (p(k)-p(k+1)) \E[X \mid \delta_{p(k+1)} < X \leq \delta_{p(k)}].
        \end{aligned}
        $
    \end{itemize}
    Plugged into the formula from above, this results in
    \begin{align*}
        \E[\text{ONL}_n] &= \sum_{i=1}^n \left(\prod_{j=1}^{i-1} p(j)\right) (n-i+1) \sum_{k=0}^{i-1} (p(k)-p(k+1)) \E[X \mid \delta_{p(k+1)} < X \leq \delta_{p(k)}]\\
        &+ \sum_{i=1}^n \left(\prod_{j=1}^{i-1} p(j)\right) \sum_{k=i}^n (p(k)-p(k+1)) \E[X \mid \delta_{p(k+1)} < X \leq \delta_{p(k)}].
    \end{align*}
    The statement of the lemma follows immediately.
\end{proof}

Combining the two lemmas we obtain a lower bound on the prophet inequality of
\[
\frac{\E[\textsc{Onl}_n]}{\E[\textsc{Opt}_n]} \geq \frac{\sum_{k=0}^n \alpha_k \E[X \mid \delta_{p(k+1)} < X \leq  \delta_{p(k)}]}{\sum_{k=0}^n \alpha_k^* \E[X \mid \delta_{p(k+1)} < X \leq \delta_{p(k)}]}\;.
\]
The expectation values in this bound heavily depend on the input distribution and cannot easily be calculated. Therefore, We require an advanced bound on the weighted mediant. We show

\begin{restatable}[Advanced lower bound on the weighted mediant]{lemma}{minbound}\label{lem:minbound}
    Let $n\in\mathbb{N}_{>0}$. For non-negative values $a_i$ and positive values $b_i$ and $w_i$ for $i\in[n]$, where the values $w_i$ are non-increasing, it holds that
\[\frac{\sum_{i=1}^n w_i a_i}{\sum_{i=1}^n w_i b_i} \geq \min_{s\in [n]} \frac{\sum_{i=1}^{s} a_i}{\sum_{i=1}^{s} b_i}.\]
\end{restatable}
\begin{proof}
Since the weights $w_i$ are positive and non-increasing in $i\in[n]$, we can find appropriate $u_s \geq 0$ for $s\in[n]$ to rewrite $w_i= \sum_{s=i}^n u_s$ with $u_n = w_n > 0$. Then
\[
\frac{\sum_{i=1}^n w_i a_i}{\sum_{i=1}^n w_i b_i} = \frac{\sum_{i=1}^n a_i\sum_{s=i}^n u_s}{\sum_{i=1}^n b_i\sum_{s=i}^n u_s} = \frac{\sum_{s=1}^n u_s \sum_{i=1}^s a_i}{\sum_{s=1}^n u_s \sum_{i=1}^s b_i} \geq \min_{s=1,...,n} \frac{\sum_{i=1}^s a_i}{\sum_{i=1}^s b_i}.
\]
The last inequality is the standard bound for weighted mediants. That is, for non-negative $a_i$ and positive $b_i$ and $w_i$, it is known that $\frac{\sum_{i=1}^n w_i a_i}{\sum_{i=1}^n w_i b_i}$ is greater or equal to the smallest of the fractions among the $a_i / b_i$.
\end{proof}

It is easy to verify that for a strictly decreasing function $p(k)$, each $\alpha_k^*$ and each expectation value in the denominator are in fact strictly positive. Consequently, the prophet inequality is bounded as follows.

\begin{equation}\label{eq:mainbound}
    \frac{\E[\textsc{Onl}_n]}{\E[\textsc{Opt}_n]} \geq \frac{\sum_{k=0}^n \alpha_k \E[X \mid \delta_{p(k+1)} < X \leq  \delta_{p(k)}]}{\sum_{k=0}^n \alpha_k^* \E[X \mid \delta_{p(k+1)} < X \leq \delta_{p(k)}]} \geq \min_{s\in \{0,\dots,n\}} \frac{\sum_{k=0}^{s} \alpha_k}{\sum_{k=0}^{s} \alpha^*_k},
\end{equation}

where the definitions of $\alpha_k^*$ and $\alpha_k$ were made in Lemma \ref{lem:alphaoff} and Lemma \ref{lem:alphaon}. Clearly, the right-hand side heavily depends on the choice of $p(k)$ that occurs in the formulas. As already stated in Theorem \ref{thm:onlbound}, we choose $p(k)=\e^{-\frac{c k}{n^2}}$ with $c=9.71$. For $n=500$, the resulting function $\frac{\sum_{k=0}^{s} \alpha_k}{\sum_{k=0}^{s} \alpha^*_k}$ is shown in Figure \ref{fig:plot1}.

\begin{figure}[tb]
    \centering
    \resizebox{0.5\textwidth}{!}{\input{images/plot1.pgf}}
    \caption{$\frac{\sum_{k=0}^{s} \alpha_k}{\sum_{k=0}^{s} \alpha^*_k}$ for $p(k)=\e^{-\frac{c k}{n^2}}$ with $c=9.71$ for $s=0,\dots,n$ with $n=500$}
    \label{fig:plot1}
\end{figure}

We are interested in the minimum of the function in Figure \ref{fig:plot1} for $n\to\infty$, since its value gives a lower bound on the prophet inequality for $n\to\infty$ through Equation (\ref{eq:mainbound}). We guess from the figure that the minimum will be attained at $s=0$ or $s$ close to $n$ at around 0.6, but for a formal proof, we give a simplified lower bound on $\frac{\sum_{k=0}^{s} \alpha_k}{\sum_{k=0}^{s} \alpha^*_k}$ in order to be able to calculate a tight bound on the minimum of the function.

Unfortunately, the expressions for $\alpha_k$ and $\alpha_k^*$ cannot be easily simplified without a loss in the resulting bound. Therefore, we calculate more involved bounds on the sums of $\alpha_k^*$ and $\alpha_k$, inducing a loss on the competitive ratio that tends to zero for $n\to\infty$, in several steps. The first bounds are presented in Lemma \ref{lem:offlinesumtos1} and Lemma \ref{lem:onlinesumtos1} in the following. In Figure \ref{fig:plot23} on the left, we see the same function as in Figure \ref{fig:plot1} and the lower bound resulting from the two lemmas for $n=500$. Although the bound does not look tight for $n=500$, we can see in Figure \ref{fig:plot23} on the right that for $n=10^5$, the lower bound from Lemma \ref{lem:offlinesumtos1} and Lemma \ref{lem:onlinesumtos1} has its minima again at around 0.6.

\begin{figure}[t]
\begin{minipage}[t]{.5\textwidth}
\resizebox{\textwidth}{!}{\input{images/plot2.pgf}}
\end{minipage}\begin{minipage}[t]{.5\textwidth}
\resizebox{\textwidth}{!}{\input{images/plot3.pgf}}
\end{minipage}
\caption{\small \begin{minipage}[t]{\textwidth}Left: $\frac{\sum_{k=0}^{s} \alpha_k}{\sum_{k=0}^{s} \alpha^*_k}$ in blue and bound from Lemmas \ref{lem:offlinesumtos1} and \ref{lem:onlinesumtos1} in orange for $s=0,\dots,n$, $n=500$\\Right: Bound from Lemmas \ref{lem:offlinesumtos1} and \ref{lem:onlinesumtos1} in orange for $s=0,\dots,n$ with $n=10^5$\end{minipage}}
\label{fig:test}
\label{fig:plot23}
\end{figure}

As stated above, we start by bounding the sum of $\alpha_k^*$ values from Lemma \ref{lem:alphaoff}, corresponding to the offline optimum.
\begin{lemma}\label{lem:offlinesumtos1}
    For all $s\in[n-2]$ and large $n$, it holds that
    \[\sum_{k=0}^s \alpha_k^* \leq n + 1 - \frac{n^2}{c(s+2)} + \left( \frac{n^2}{c(s+2)} - \frac{c}{2} - 1 \right) \e^{-\frac{c(s+2)}{n}} \pm o(1).\]
    For $s\in\{n-1,n\}$,
    \[\sum_{k=0}^s \alpha_k^* \leq n \pm o(1).\]
\end{lemma}

\begin{proof}
    We begin by summing the terms in Lemma \ref{lem:alphaoff} for arbitrary $s\in[n]$.
    \begin{align*}
        &\sum_{k=0}^s \alpha_k^* \leq (1-p(1))\frac{n(n+1)}{2} + \sum_{i=2}^n p(1)^i-p(2)^i-i \cdot p(2)^{i-1} (p(1)-p(2))\\
        &~~~ + \sum_{i=1}^n i \cdot p(2)^{i-1} (p(1)-p(2)) + \sum_{i=2}^n p(2)^i - p(3)^i - i \cdot p(3)^{i-1} (p(2)-p(3))\\
        &~~~ + \sum_{i=1}^n i \cdot p(3)^{i-1} (p(2)-p(3)) + \sum_{i=2}^n p(3)^i - p(4)^i - i \cdot p(4)^{i-1} (p(3)-p(4))\\
        &~~~ + \dots\\
        &~~~ + \sum_{i=1}^n i \cdot p(s+1)^{i-1} (p(s)-p(s+1))\\
        &~~~~~~~~~~~~~ + \sum_{i=2}^n p(s+1)^i - p(s+2)^i - i \cdot p(s+2)^{i-1} (p(s+1)-p(s+2))\\
        &= (1-p(1))\frac{n(n+1)}{2} + \sum_{i=2}^n p(1)^i + (p(1)-p(2)) + (p(2)-p(3)) + \dots + (p(s)-p(s+1))\\
        &~~~ + \sum_{i=2}^n -p(s+2)^i - i\cdot p(s+2)^{i-1}(p(s+1)-p(s+2))\\
        &= (1-p(1)) \frac{n(n+1)}{2} + \sum_{i=1}^n p(1)^i -p(s+1) - \sum_{i=2}^n p(s+2)^i + i\cdot p(s+2)^{i-1}(p(s+1)-p(s+2)).
    \end{align*}
    Replacing $i-1$ in the exponent in the latter sum by $n$ in the formula and rewriting the sums in closed form yields
    \begin{align}\label{eq:nodistinct}
    \begin{split}
         \sum_{k=0}^s \alpha_k^* &\leq (1-p(1))\frac{n(n+1)}{2} + \frac{p(1)(1-p(1)^n)}{1-p(1)} - p(s+1) - \frac{p(s+2)(p(s+2)-p(s+2)^n)}{1-p(s+2)}\\
         &~~~ - \frac{n^2+n-2}{2} p(s+2)^n (p(s+1)-p(s+2)).
    \end{split}
    \end{align}
    From here on, we do a case distinction and first assume that $s\in\{n-1,n\}$. In this case, we apply $p(s+2)=0$ and omit $-p(s+1)\leq0$. Further applying $p(1)=\e^{-\frac{c}{n^2}}$ and then $\e^{-x} \geq 1-x$ yields
    \begin{align*}
         \sum_{k=0}^s \alpha_k^* &\leq \left(1-\e^{-\frac{c}{n^2}}\right) \frac{n(n+1)}{2} + \frac{\e^{-\frac{c}{n^2}}(1-\e^{-\frac{c}{n}})}{1-\e^{-\frac{c}{n^2}}} \leq \frac{c}{n}\frac{n+1}{2} + \frac{\e^{-\frac{c}{n^2}}(1-\e^{-\frac{c}{n}})}{1-\e^{-\frac{c}{n^2}}}.
    \end{align*}
    Further using that $\e^{-\frac{c}{n^2}} \leq 1-\frac{c}{n^2}+\frac{c^2}{2n^4}$ and $\e^{-\frac{c}{n}} \geq 1-\frac{c}{n}+\frac{c^2}{2n^2}-O\left(\frac{1}{n^3}\right)$ implies
    \begin{align*}
         \sum_{k=0}^s \alpha_k^* &\leq \frac{c}{n}\frac{n+1}{2} + \frac{(1-\frac{c}{n^2}+\frac{c^2}{2n^4})(\frac{c}{n}-\frac{c^2}{2n^2}+O\left(\frac{1}{n^3}\right))}{\frac{c}{n^2}-\frac{c^2}{2n^4}}.
    \end{align*}
    Note that $\left(1-\frac{c}{n^2}+\frac{c^2}{2n^4}\right) \left(\frac{c}{n}-\frac{c^2}{2n^2}+O\left(\frac{1}{n^3}\right)\right) = \frac{c}{n} - \frac{c^2}{2n^2} \pm O\left(\frac{1}{n^3}\right)$ and $1 / (\frac{c}{n^2}-\frac{c^2}{2n^4}) = \frac{2n^4}{2n^2c-c^2} = \frac{n^2}{c} + \frac{c}{4n^2-2c} + \frac{1}{2} = \frac{n^2}{c} + \O(1)$, such that the corresponding fraction from the earlier formula equals $n-\frac{c}{2} \pm O\left(\frac{1}{n}\right)$, such that the overall sum can be bounded as
    \begin{align*}
         \sum_{k=0}^s \alpha_k^* &\leq \frac{c}{2} + n - \frac{c}{2} \pm o(1) = n \pm o(1),
    \end{align*}
finishing the proof for the case $s\in\{n-1,n\}$. For $s\in[n-2]$, we can use Equation \eqref{eq:nodistinct} again and simplify the first terms as in the previous case to derive
	\begin{align*}
         \sum_{k=0}^s \alpha_k^* &\leq n \pm o(1) - p(s+1) - \frac{p(s+2)(p(s+2)-p(s+2)^n)}{1-p(s+2)}\\
         &~~~ - \frac{n^2+n-2}{2} p(s+2)^n (p(s+1)-p(s+2)).
    \end{align*}
    Next, we claim that it holds for all $s\leq n-2$ that
    \begin{equation}\label{eq:pdif}
        p(s+1)-p(s+2) \geq \left(1 - \frac{c}{n}\right) \frac{c}{n^2}.
    \end{equation}
    For a proof of the claim, observe that the function $f(k) := p(k)-p(k+1)$ is decreasing in $k$ for $k\leq n-1$ for the choice $p(k)=\e^{-\frac{c k}{n^2}}$ with $c=9.71$: Its derivative is $-\frac{c}{n^2}f(k)$, and since $p(k)$ is obviously strictly decreasing in $k$, this is negative. It follows that $f(k) \geq f(n-1)$ for all $k \leq n-1$. Applying appropriate Taylor-expansions yields that for all $k \leq n-1$,
    \[f(k) \geq f(n-1) = p(n-1) - p(n) = \e^{-\frac{c}{n}} \left(\e^{\frac{c}{n^2}} - 1 \right) \geq \left(1 - \frac{c}{n}\right) \left(1+\frac{c}{n^2}-1\right) = \left(1 - \frac{c}{n}\right) \frac{c}{n^2}.\]
    Consequently, we use Equation \eqref{eq:pdif} to deduce that
    \begin{align*}
         \sum_{k=0}^s \alpha_k^* &\leq n  \pm o(1) - p(s+1) - \frac{p(s+2)(p(s+2)-p(s+2)^n)}{1-p(s+2)} - \frac{n^2+n-2}{2} p(s+2)^n \left(1 - \frac{c}{n}\right) \frac{c}{n^2}\\
         &= n - p(s+1) - \frac{p(s+2)(p(s+2)-p(s+2)^n)}{1-p(s+2)} - \frac{c}{2} p(s+2)^n \pm o(1),
    \end{align*}
    where applying $p(k)=\e^{-\frac{c k}{n^2}}$ yields
    \begin{align*}
         \sum_{k=0}^s \alpha_k^* &\leq n - \e^{-\frac{c(s+1)}{n^2}} - \frac{\e^{-\frac{c(s+2)}{n^2}}}{1-\e^{-\frac{c(s+2)}{n^2}}}\left(\e^{-\frac{c(s+2)}{n^2}}-\e^{-\frac{c(s+2)}{n}}\right) - \frac{c}{2} \e^{-\frac{c(s+2)}{n}} \pm o(1),
    \end{align*}
    where applying $\e^{-x} \geq 1-x$ at some points yields
    \begin{align*}
         \sum_{k=0}^s \alpha_k^* &\leq n - 1 +\frac{c(s+1)}{n^2} - \frac{n^2}{c(s+2)} \left(1-\frac{c(s+2)}{n^2}\right) \left(1-\frac{c(s+2)}{n^2}-\e^{-\frac{c(s+2)}{n}}\right) - \frac{c}{2} \e^{-\frac{c(s+2)}{n}} \pm o(1)\\
         &= n - 1 - \frac{n^2}{c(s+2)} \left(1-\frac{c(s+2)}{n^2}\right) \left(1-\frac{c(s+2)}{n^2}-\e^{-\frac{c(s+2)}{n}}\right) - \frac{c}{2} \e^{-\frac{c(s+2)}{n}} \pm o(1)\\
         &= n - 1 + \left(1-\frac{n^2}{c(s+2)}\right) \left(1-\frac{c(s+2)}{n^2}-\e^{-\frac{c(s+2)}{n}}\right) - \frac{c}{2} \e^{-\frac{c(s+2)}{n}} \pm o(1)\\
         &= n - 1 + 1-\frac{n^2}{c(s+2)} - \frac{c(s+2)}{n^2} + 1 - \left(1-\frac{n^2}{c(s+2)} + \frac{c}{2}\right) \e^{-\frac{c(s+2)}{n}} \pm o(1)\\
         &\leq n + 1 - \frac{n^2}{c(s+2)} + \left( \frac{n^2}{c(s+2)} - \frac{c}{2} - 1 \right) \e^{-\frac{c(s+2)}{n}} \pm o(1),
    \end{align*}
    finishing the proof for the case $s\in[n-2]$.
\end{proof}

Analogous to the previous lemma, we give a bound on the sum of $\alpha_k$ values from Lemma \ref{lem:alphaon}, corresponding to the online optimum.

\begin{lemma}\label{lem:onlinesumtos1}
For all $s=0,\dots,n$ and $n\geq\sqrt{c}$,
\begin{align*}
    \sum_{k=0}^s \alpha_k &\geq \left(1 - \frac{c}{n}\right) \left( (s+1) \bigg( \sqrt{\frac{c\pi}{2}} \erf{\sqrt{\frac{c}{2}}} - \frac{c}{n} + \e^{-\frac{c}{2}} \right)\\
    &\quad\quad\quad\quad\quad\quad + n - \sqrt{\frac{\pi}{2c}} (n + c (s+1)) \erf{\sqrt{\frac{c}{2}}\frac{s}{n}} - (n+1) \e^{-\frac{c s^2}{2 n^2}} \bigg).
\end{align*}
\end{lemma}

\begin{proof}
    First, recall from Equation \eqref{eq:pdif} that it holds for all $k\leq n-1$ that
    \[p(k)-p(k+1) \geq \left(1 - \frac{c}{n}\right) \frac{c}{n^2}.\]
    Also note that $p(n)-p(n+1) = p(n) - 0 = \e^{-\frac{c}{n}} \geq 1-\frac{c}{n} \geq \left(1 - \frac{c}{n}\right) \frac{c}{n^2}$ for $n\geq\sqrt{c}$. The imprecision of this bound does not hurt our overall result. Applying this to Lemma \ref{lem:alphaon} yields for large values of $n$ and all $s=0,\dots,n$ that
    \begin{align*}
        \alpha_k &\geq \left(1 - \frac{c}{n}\right) \frac{c}{n^2}\left(\sum_{i=1}^k \prod_{j=1}^{i-1} p(j) + \sum_{i=k+1}^n (n-i+1) \prod_{j=1}^{i-1} p(j)\right)\\
        &\geq \left(1 - \frac{c}{n}\right) \frac{c}{n^2}\left(\sum_{i=1}^k \e^{-\frac{c i^2}{2n^2}} + \sum_{i=k+1}^n (n-i+1) \e^{-\frac{c i^2}{2n^2}}\right).
    \end{align*}
    Consequently, for all $s=0,\dots,n$,
    \begin{align*}
        \sum_{k=0}^s \alpha_k &\geq \left(1 - \frac{c}{n}\right) \frac{c}{n^2} \sum_{k=0}^s \left(\sum_{i=1}^k \e^{-\frac{c i^2}{2n^2}} + \sum_{i=k+1}^n (n-i+1) \e^{-\frac{c i^2}{2n^2}}\right)\\
        &= \left(1 - \frac{c}{n}\right) \frac{c}{n^2} \left(\sum_{k=1}^s ((s+1-k)+k(n-k+1)) \e^{-\frac{c k^2}{2n^2}} + \sum_{k=s+1}^n (s+1)(n-k+1) \e^{-\frac{c k^2}{2n^2}}\right)\\
        &= \left(1 - \frac{c}{n}\right) \frac{c}{n^2}(s+1) \left(\sum_{k=1}^s \left(\frac{k(n-k)}{s+1}+1\right) \e^{-\frac{c k^2}{2n^2}} + \sum_{k=s+1}^n (n-k+1) \e^{-\frac{c k^2}{2n^2}}\right)\\
        &= \left(1 - \frac{c}{n}\right) \frac{c}{n^2}(s+1) \left(\sum_{k=1}^s \left(\frac{k(n-k)}{s+1}-(n-k)\right) \e^{-\frac{c k^2}{2n^2}} + \sum_{k=1}^n (n-k+1) \e^{-\frac{c k^2}{2n^2}}\right)\\
        &= \left(1 - \frac{c}{n}\right) \frac{c}{n^2}(s+1) \left(\sum_{k=1}^n (n-k+1) \e^{-\frac{c k^2}{2n^2}} - \sum_{k=1}^s (n-k)\left(\frac{s-k+1}{s+1}\right) \e^{-\frac{c k^2}{2n^2}}\right)\\
        &= \left(1 - \frac{c}{n}\right) \frac{c}{n^2} \left((s+1)\sum_{k=1}^n (n-k+1) \e^{-\frac{c k^2}{2n^2}} - \sum_{k=1}^s (n-k)(s-k+1) \e^{-\frac{c k^2}{2n^2}}\right).
    \end{align*}
    Here, the summands in both sums are monotonically decreasing in $k$, so that we can approximate them by using integrals.
    \begin{align*}
        \sum_{k=0}^s \alpha_k &\geq \left(1 - \frac{c}{n}\right) \frac{c}{n^2} \left((s+1)\int_{1}^{n} (n-k+1) \e^{-\frac{c k^2}{2n^2}}\,dk - \int_{0}^s (n-k)(s-k+1) \e^{-\frac{c k^2}{2n^2}}\,dk\right)\\
        &= \left(1 - \frac{c}{n}\right) \frac{c}{n^2} \bigg((s+1) \frac{n^2}{c} \left(\frac{n+1}{n} \sqrt{\frac{\pi c}{2}} \left( \erf{\sqrt{\frac{c}{2}}} - \erf{\sqrt{\frac{c}{2}}\frac{1}{n}} \right) + \left(\e^{-\frac{c}{2}} - \e^{-\frac{c}{2n^2}}\right)  \right)\\
        &~~~~~~~~~ - \frac{n^2}{c} \left( \sqrt{\frac{\pi}{2c}} (c(s+1)+n) \erf{\sqrt{\frac{c}{2}}\frac{s}{n}} + (n+1)\e^{-\frac{c s^2}{2n^2}} -n-(s+1) \right) \bigg)\\
        &= \left(1 - \frac{c}{n}\right) \bigg( (s+1) \left( \frac{n+1}{n} \sqrt{\frac{\pi c}{2}} \left( \erf{\sqrt{\frac{c}{2}}} - \erf{\sqrt{\frac{c}{2}}\frac{1}{n}} \right) + \e^{-\frac{c}{2}} - \e^{-\frac{c}{2n^2}} + 1 \right)\\
        &~~~ + n - \sqrt{\frac{\pi}{2c}} (c(s+1)+n) \erf{\sqrt{\frac{c}{2}}\frac{s}{n}} - (n+1)\e^{-\frac{c s^2}{2n^2}} \bigg).
    \end{align*}
    Note that $n \leq n+1$ and that the part multiplied with it is greater than 0. Also, $\e^{-\frac{c}{2n^2}} \leq 1$. Further, Taylor-expansion of $\erf{x}$ in $x=0$ yields that $\erf{x} \leq (2x)/\sqrt{\pi}$ for all $x\geq0$. Applying this expansion at $x = \sqrt{\frac{c}{2}}\frac{1}{n}$ as well as the other bounds implies the result.
\end{proof}

What makes the further analysis complicated is that, depending on the value of $s$, different terms of the respective sums in our lower bound dominate their value. Also, for this simplified bound it is still not possible to easily calculate the minima. Thus, in Lemma \ref{lem:offlinesumtos2} and Lemma \ref{lem:onlinesumtos2}, following below, we calculate lower bounds of simple functions, e.g., small-degree polynomials in $s$. For a sufficiently small constant $\varepsilon$, we consider the cases where $s$ is in any of the following intervals $[0, \varepsilon n]$, $[\varepsilon n, \frac{n}{4}]$, $[\frac{n}{4}, \frac{n}{2}]$, $[\frac{n}{2},\frac{3n}{4}]$, $[\frac{3n}{4}, n-2]$, or if it is $n-1$ or $n$. Then the prophet inequality is at least the minimum over all cases. Plots of the formulas of the bounds for the different ranges of $s$ given $n=10^5$ can be seen in Figure \ref{fig:plot4}.

\begin{figure}[t]
    \centering
    \resizebox{0.5\textwidth}{!}{\input{images/plot4.pgf}}
    \caption{Plots of the bound from Lemmas \ref{lem:offlinesumtos1} and \ref{lem:onlinesumtos1} in orange and bound from Lemmas \ref{lem:offlinesumtos2} and \ref{lem:onlinesumtos2} in green for $s=0,\dots,n$ with $n=10^5$}
    \label{fig:plot4}
\end{figure}

We begin by providing bounds on the sum over the $\alpha^*_k$ values from Lemma \ref{lem:alphaoff}, corresponding to the offline optimum, on restricted ranges of $s$ using simple functions.

\begin{lemma}\label{lem:offlinesumtos2}
\begin{align*}
&\text{Let $\varepsilon>0$ be a small constant. Then, }\alpha_0^* \leq \frac{c}{2} + o(1).\\
&\text{For $0 < s < \varepsilon n$, define $\varepsilon^{\prime\prime} = \varepsilon \frac{c(c+2)}{2}$, then }\sum_{k=0}^s \alpha_k^* \leq \frac{c(s+1)}{2} + \varepsilon^{\prime\prime} \pm o(1).\\
&\text{For all $\varepsilon n \leq s \leq \frac{1}{4}n$, }\sum_{k=0}^s \alpha_k^* \leq \frac{c}{2}s - \frac{64-16c+2c^2-64\e^{-\frac{c}{4}}}{n c}s^2 \pm \O(1).\\
&\text{For all $\frac{1}{4}n \leq s \leq n-2$, }\sum_{k=0}^s \alpha_k^* \leq n - \frac{n^2}{c s} \left(1-\e^{-\frac{c}{4}}\right) \pm \O(1).\\
&\text{For $s\in\{n-1,n\}$, }\sum_{k=0}^{s} \alpha_k^* \leq n \pm \O(1).\\
\end{align*}
\end{lemma}

\begin{proof}
    First, consider $s=0$, where we know from Lemma \ref{lem:alphaoff} that
    \begin{align*}
        \alpha_0^* &= \left(1-\e^{-\frac{c}{n^2}}\right)\frac{n(n+1)}{2} + \sum_{i=2}^n \e^{-\frac{c i}{n^2}}-\e^{-\frac{2c i}{n^2}}-i \cdot \e^{-\frac{2c(i-1)}{n^2}} \left(\e^{-\frac{c}{n^2}}-\e^{-\frac{2c}{n^2}}\right).
    \end{align*}
    Replacing the $i-1$ by $n$ and applying appropriate Taylor-approximations on $e^{-x}$ in $x=0$ yields that
    \begin{align*}
        \alpha_0^* &\leq \frac{c}{n^2}\frac{n(n+1)}{2} + \sum_{i=2}^n 1-\frac{ci}{n^2}+\frac{c^2i^2}{2n^4}-1+\frac{2ci}{n^2}-i \cdot \left(1-\frac{2c}{n}\right)\left(1 - \frac{c}{n^2} - 1 + \frac{2c}{n^2} - \frac{c^2}{n^4}\right)\\
        &= \frac{c}{2}\frac{n+1}{n} + \sum_{i=2}^n \frac{ci}{n^2}+\frac{c^2i^2}{2n^4} - i \cdot \left(1-\frac{2c}{n}\right)\left(\frac{c}{n^2} - \frac{c^2}{n^4}\right)\\
        &= \frac{c}{2}\frac{n+1}{n} + \sum_{i=2}^n \frac{c^2i^2}{2n^4} - \frac{c^2i}{n^4} + \frac{2c^2i}{n^3} - \frac{2c^3i}{n^5} \leq \frac{c}{2}\frac{n+1}{n} + \sum_{i=2}^n \frac{c^2i^2}{2n^4} + \frac{2c^2i}{n^3}\\
        &\leq \frac{c}{2}\frac{n+1}{n} + \sum_{i=2}^n \frac{c^2}{2n^2} + \frac{2c^2}{n^2} = \frac{c}{2}\frac{n+1}{n} + (n-1) \frac{5c^2}{2n^2} = \frac{c}{2} + o(1).
    \end{align*}
    Second, consider $1\leq s < \varepsilon n$ for small $\epsilon>0$. In this case, we use the inequality $\e^{-\frac{c(s+2)}{n}} \leq 1-\frac{c(s+2)}{n}+\frac{c^2(s+2)^2}{2n^2}$ on Lemma \ref{lem:offlinesumtos1} in order to see that
    \begin{align*}
         \sum_{k=0}^s \alpha_k^* &\leq n + 1 - \frac{n^2}{c(s+2)} + \left( \frac{n^2}{c(s+2)} - \frac{c}{2} - 1 \right) \left(1-\frac{c(s+2)}{n}+\frac{c^2(s+2)^2}{2n^2}\right) + o(1),
    \end{align*}
    where
    \begin{align*}
         &\left( \frac{n^2}{c(s+2)} - \frac{c}{2} - 1 \right) \left(1-\frac{c(s+2)}{n}+\frac{c^2(s+2)^2}{2n^2}\right)\\
         &~~= \frac{n^2}{c(s+2)} - \frac{c}{2} - 1 - n + \frac{c^2(s+2)}{2n} +\frac{c(s+2)}{n} + \frac{c(s+2)}{2} - \frac{c^3(s+2)^2}{4n^2} - \frac{c^2(s+2)^2}{2n^2}\\
         &~~= \frac{n^2}{c(s+2)} + \frac{c(s+1)}{2} - 1 - n + \frac{c^2s}{2n} +\frac{c s}{n} - \frac{c^3 s^2}{4n^2} - \frac{c^2 s^2}{2n^2} \pm o(1)\\
         &~~\leq \frac{n^2}{c(s+2)} + \frac{c(s+1)}{2} - 1 - n + \varepsilon \frac{c(c+2)}{2} \pm o(1).
    \end{align*}
    Plugging this into the earlier formula yields the result.\\
    Third, consider $\varepsilon n \leq s \leq \frac{1}{4}n$. Recall from Lemma \ref{lem:offlinesumtos1}, that
    \begin{align*}
        \sum_{k=0}^s \alpha_k^* &\leq n + 1 - \frac{n^2}{c(s+2)} + \left( \frac{n^2}{c(s+2)} - \frac{c}{2} - 1 \right) \e^{-\frac{c(s+2)}{n}} \pm o(1)\\
        &= n + 1 - \frac{n^2}{c(s+2)} \left(1 - \e^{-\frac{c(s+2)}{n}}\right) - \left(\frac{c}{2} + 1\right) \e^{-\frac{c(s+2)}{n}} \pm o(1)\\
        &\leq n + 1 - \frac{n^2}{c(s+2)} \left(1 - \e^{-\frac{c(s+2)}{n}} \right) \pm o(1)\\
        &= n + 1 - \left( \frac{n^2}{c s} + \frac{n^2}{c(s+2)} - \frac{n^2}{c s}\right)\left(1 - \e^{-\frac{c(s+2)}{n}}\right) \pm o(1)\\
        &= n + 1 - \left( \frac{n^2}{c s} - \frac{2n^2}{c s(s+2)} \right)\left(1 - \e^{-\frac{c(s+2)}{n}}\right) \pm o(1)\\
        &= n - \frac{n^2}{c s} \left(1 - \e^{-\frac{c(s+2)}{n}}\right) \pm \O(1) \leq n - \frac{n^2}{c s} \left(1 - \e^{-\frac{c s}{n}}\right) \pm \O(1).
    \end{align*}
    Next, we define $g(s)$ to be a cubic function fulfilling to cross $\e^{-\frac{c s}{n}}$ at $s=0$ and $s=\frac{1}{4}n$ and to have the same slope and curvature at $s=0$, obtaining $g(s) = 1 - \frac{c}{n}s + \frac{c^2}{2n^2}s^2 - \frac{64-16c+2c^2-64\e^{-\frac{c}{4}}}{n^3}s^3$. We show that $g(s) \geq \e^{-\frac{c s}{n}}$ in the interval of interest by defining $\hat{g}(s) = g(s) - \e^{-\frac{c s}{n}}$ and showing $\hat{g}(s) \geq 0$ for all $s \in [0,\frac{1}{4}n]$. By construction, it holds that $\hat{g}(0) = \hat{g}^\prime(0) = \hat{g}^{\prime\prime}(0) = 0$. Further, one can easily calculate $\hat{g}^{\prime\prime\prime}(s) = \frac{1}{n^3} \left( c^3 \e^{-\frac{c s}{n}} - z \right)$, where $z = 6\left( 64-16c+2c^2-64\e^{-\frac{1}{4}c} \right)$. Due to the monotonicity of the exponential function, $\hat{g}^{\prime\prime\prime}$ takes its lowest value at $s=\frac{1}{4}n$, but one can easily check that this is still positive for $c=9.71$. Consequently, the change of the curvature of $g$ is always positive over $[0,\frac{1}{4}n]$. Since $\hat{g}^{\prime\prime}(0) = 0$, this means that $g$ is convex in the whole interval, and due to $g(0) = 0$, this shows the claim. Plugging it into the earlier formula, we get that
    \begin{align*}
        \sum_{k=0}^s \alpha_k^* &\leq n - \frac{n^2}{c s} \left(1 - 1 + \frac{c}{n}s - \frac{c^2}{2n^2}s^2 + \frac{64-16c+2c^2-64\e^{-\frac{c}{4}}}{n^3}s^3\right) \pm \O(1)\\
        &= \frac{c}{2}s - \frac{64-16c+2c^2-64\e^{-\frac{c}{4}}}{n c}s^2 \pm \O(1).
    \end{align*}
    Fourth, consider $\frac{1}{4}n \leq s \leq n-2$. Analogous calculations to the previous case and observing that $\e^{-\frac{c s}{n}} \leq \e^{-\frac{c}{4}}$ directly implies the result.\\
    Fifth, consider $s\in\{n-1,n\}.$ The result follows directly from Lemma \ref{lem:offlinesumtos1}.
\end{proof}

We now proceed analogously for the sum over the $\alpha_k$ values from Lemma \ref{lem:alphaon}, corresponding to the online optimum, deriving simple bounds on restricted ranges of $s$.

\begin{lemma}\label{lem:onlinesumtos2}
Let $\varepsilon >0$ be a small constant and define
\begin{align*}
&\textstyle{b_1 = \sqrt{\frac{c\pi}{2}} \erf{\sqrt{\frac{c}{2}}} + \e^{-\frac{c}{2}} - 1 \approx 2.906,}\\
&\textstyle{b_2 = \sqrt{\frac{c\pi}{2}} \erf{\sqrt{\frac{c}{2}}} + \e^{-\frac{c}{2}} + 4 \e^{-\frac{c}{8}} \left(\e^{\frac{3c}{32}} - 1\right) - \e^{-0.85^2} - c \sqrt{\frac{\pi}{2c}} \left(\erf{0.85} - \frac{1.7 \e^{-0.85^2}}{\sqrt{\pi}}\right) \approx 3.994,}\\
&\textstyle{b_3 = 1 - \e^{-\frac{c}{8}} \left(2\e^{\frac{3c}{32}} - 1\right) - \sqrt{\frac{\pi}{2c}} \left(\erf{0.85} - \frac{1.7 \e^{-0.85^2}}{\sqrt{\pi}}\right) \approx -0.302,}\\
&\textstyle{b_4 = \sqrt{\frac{c\pi}{2}} \erf{\sqrt{\frac{c}{2}}} + \e^{-\frac{c}{2}} + 4 \e^{-\frac{9c}{32}} \left(\e^{\frac{5c}{32}} - 1\right) - \e^{-1.35^2} - c \sqrt{\frac{\pi}{2c}} \left(\erf{1.35} - \frac{2.7 \e^{-1.35^2}}{\sqrt{\pi}}\right) \approx 1.948,}\\
&\textstyle{b_5 = 1 - \e^{-\frac{9c}{32}} \left(3\e^{\frac{5c}{32}} - 2\right) - \sqrt{\frac{\pi}{2c}} \left(\erf{1.35} - \frac{2.7 \e^{-1.35^2}}{\sqrt{\pi}}\right) \approx -0.041,}\\
&\textstyle{b_6 = \sqrt{\frac{c\pi}{2}} \erf{\sqrt{\frac{c}{2}}} + \e^{-\frac{c}{2}} + 4 \e^{-\frac{c}{2}} \left(\e^{\frac{7c}{32}} - 1\right) - c \sqrt{\frac{\pi}{2c}} \erf{\sqrt{\frac{c}{2}}} \approx 0.237,}\\
&\textstyle{b_7 = 1 - \e^{-\frac{c}{2}} \left(4\e^{\frac{7c}{32}} - 3\right) - \sqrt{\frac{\pi}{2c}} \erf{\sqrt{\frac{c}{2}}} \approx 0.361 \text{ and}}\\
&\textstyle{\varepsilon^\prime = c \varepsilon \left( \sqrt{\frac{c\pi}{2}} \erf{\sqrt{\frac{c}{2}}} + \e^{-\frac{c}{2}} \right) \approx 37.928\varepsilon.}\\
&\text{Then, for $s < \varepsilon n$, }\sum_{k=0}^s \alpha_k \geq (s+1)  b_1 - \varepsilon^\prime - o(1).\\
&\text{For $\varepsilon n \leq s \leq \frac{1}{4}n$, }\sum_{k=0}^s \alpha_k \geq s \left( b_1 - \frac{c s}{n} \frac{64+c}{128} \right) \pm \O(1).\\
&\text{For $\frac{1}{4}n < s \leq \frac{1}{2}n$, }\sum_{k=0}^s \alpha_k \geq s \left(b_2 - \frac{c s}{n} \e^{-0.85^2} \right) + n b_3 \pm \O(1).\\
&\text{For $\frac{1}{2}n < s \leq \frac{3}{4}n$, }\sum_{k=0}^s \alpha_k \geq s \left( b_4 - \frac{c s}{n} \e^{-1.35^2} \right) + n b_5 \pm \O(1).\\
&\text{For $\frac{3}{4}n < s \leq n$, }\sum_{k=0}^s \alpha_k \geq s b_6 + n b_7 \pm \O(1).
\end{align*}	    
\end{lemma}

\begin{proof}
    The structure of the proof is as follows. In all cases, we will use Lemma \ref{lem:onlinesumtos1}, which says for all $s=0,\dots,n$ and $n\geq\sqrt{c}$ that
    \begin{align*}
    \sum_{k=0}^s \alpha_k &\geq \left(1 - \frac{c}{n}\right) \left( (s+1) \bigg( \sqrt{\frac{c\pi}{2}} \erf{\sqrt{\frac{c}{2}}} - \frac{c}{n} + \e^{-\frac{c}{2}} \right)\\
    &\quad\quad\quad\quad\quad\quad + n - \sqrt{\frac{\pi}{2c}} (n + c (s+1)) \erf{\sqrt{\frac{c}{2}}\frac{s}{n}} - (n+1) \e^{-\frac{c s^2}{2 n^2}} \bigg).
\end{align*}
    Then, in all five cases, we will
    \begin{enumerate}
        \item upper bound $\erf{x}$ by some linear function $f\in\{f_1,\dots,f_5\}$ and then apply it for $x=\sqrt{\frac{c}{2}}\frac{s}{n}$,
        \item upper bound $\e^{-\frac{c x^2}{2 n^2}}$ by some polynomial function $g\in\{g_1,\dots,g_5\}$ and then apply it for $x=s$,
        \item simplify.
    \end{enumerate}
    For an overview, we will next specify all of the functions described above. For approximating the error function, generally note that a linear function $a x + b$ that crosses $\erf{x}$ in $x=d$ and has the same slope in $x=d$ fulfills $a=\frac{2\e^{-d^2}}{\sqrt{\pi}}$ and $b=\erf{d}-\frac{2d\e^{-d^2}}{\sqrt{\pi}}$. Since $\erf{x}$ is concave for $x \geq 0$, such linear functions are upper bounds on $\erf{x}$ by construction. For the first two cases, we choose $d=0$, for the third case $d=0.85$ and for the fourth case $d=1.35$, so we upper bound $\erf{x}$ by
    \begin{align*}
        f_1(x) &= f_2(x) = \frac{2}{\sqrt{\pi}}x,\\
        f_3(x) &= \frac{2\e^{-0.85^2}}{\sqrt{\pi}} x + \erf{0.85} - \frac{1.7 \e^{-0.85^2}}{\sqrt{\pi}},\\
         f_4(x) &= \frac{2\e^{-1.35^2}}{\sqrt{\pi}} x + \erf{1.35} - \frac{2.7 \e^{-1.35^2}}{\sqrt{\pi}}.
    \end{align*}
    Lastly, since $\erf{x}$ is monotonically increasing and the greatest $x$ we ever apply it for is $\sqrt{\frac{c}{2}}$, we choose
    \[f_5(x) = \erf{\sqrt{\frac{c}{2}}}.\]
    For approximating $\e^{-\frac{c x^2}{2 n^2}}$ in the first case, we choose the trivial upper bound $g_1(x) = 1$.
    Second, we choose the Taylor-expansion in $x=0$ of the next higher degree that is an upper bound,
    \[g_2(x) = 1 - \frac{c x^2}{2n^2} + \frac{c^2x^4}{8n^4}.\]
    In the third case, we choose the linear function which crosses $\e^{-\frac{c x^2}{2 n^2}}$ in $x=\frac{1}{4}n$ and $x=\frac{1}{2}n$,
    \[g_3(x) = -\frac{4 \e^{-\frac{c}{8}} \left(\e^{\frac{3c}{32}} - 1\right)}{n} x + \e^{-\frac{c}{8}} \left(2\e^{\frac{3c}{32}} - 1\right).\]
    In order to see that this is in fact an upper bound between $x=\frac{1}{4}n$ and $x=\frac{1}{2}n$, note that $\e^{-\frac{c x^2}{2 n^2}}$ has exactly one inflection point for $x \geq 0$, which is at $x = n/\sqrt{c} \in (\frac{1}{4}n, \frac{1}{2}n)$. Before this point, the function is concave and afterwards it is convex. Further, the slope of $\e^{-\frac{c x^2}{2 n^2}}$ in $x=\frac{1}{4}n$ is steeper than the slope of $g_3(x)$. Consequently, it has exactly two intersections with $g_3(x)$ which we know by construction and we can conclude our claim.\\
    The last two cases will have two linear functions as upper bounds on $\e^{-\frac{c x^2}{2 n^2}}$ that are constructed analogously to the previous case. They are upper bounds by construction because we have already seen that $\e^{-\frac{c x^2}{2 n^2}}$ is convex in the interval starting from $x=\frac{1}{2}n$. The remaining functions are
    \begin{align*}
        g_4(x) &= -\frac{4 \e^{-\frac{9c}{32}} \left(\e^{\frac{5c}{32}} - 1\right)}{n} x + \e^{-\frac{9c}{32}} \left(3\e^{\frac{5c}{32}} - 2\right),\\
        g_5(x) &= -\frac{4 \e^{-\frac{c}{2}} \left(\e^{\frac{7c}{32}} - 1\right)}{n} x + \e^{-\frac{c}{2}} \left(4\e^{\frac{7c}{32}} - 3\right).
    \end{align*}
    Now, for the first case, $s < \varepsilon n$, we use $\erf{\sqrt{\frac{c}{2}}\frac{s}{n}} \leq f_1(\sqrt{\frac{c}{2}}\frac{s}{n}) = \frac{2}{\sqrt{\pi}} \sqrt{\frac{c}{2}}\frac{s}{n}$ and $\e^{-\frac{c s^2}{2n^2}} \leq g_1(s) = 1$, so
    \begin{align*}
        \sum_{k=0}^s \alpha_k &\geq \left(1 - \frac{c}{n}\right) \bigg( (s+1) \left( \sqrt{\frac{c\pi}{2}} \erf{\sqrt{\frac{c}{2}}} - \frac{c}{n} + \e^{-\frac{c}{2}} \right)\\
        &\quad\quad\quad\quad\quad\quad + n - \sqrt{\frac{\pi}{2c}} (n + c (s+1)) \frac{2}{\sqrt{\pi}} \sqrt{\frac{c}{2}}\frac{s}{n} - (n+1) \cdot 1 \bigg)\\
        &= \left(1 - \frac{c}{n}\right) \left( (s+1) \left( \sqrt{\frac{c\pi}{2}} \erf{\sqrt{\frac{c}{2}}} - \frac{c}{n}  + \e^{-\frac{c}{2}} \right) - 1 - (n + c (s+1)) \frac{s}{n} \right)\\
        &= \left(1 - \frac{c}{n}\right) (s+1) \left( \sqrt{\frac{c\pi}{2}} \erf{\sqrt{\frac{c}{2}}} - \frac{c(s+1)}{n} + \e^{-\frac{c}{2}} - 1\right)\\
        &= (s+1) \left( \sqrt{\frac{c\pi}{2}} \erf{\sqrt{\frac{c}{2}}} + \e^{-\frac{c}{2}} - 1\right) - \frac{c}{n} (s+1) \left( \sqrt{\frac{c\pi}{2}} \erf{\sqrt{\frac{c}{2}}} + \e^{-\frac{c}{2}} - \frac{c(s+1)}{n} \right)\\
        &\geq (s+1) \left( \sqrt{\frac{c\pi}{2}} \erf{\sqrt{\frac{c}{2}}} + \e^{-\frac{c}{2}} - 1\right) - c \varepsilon \left( \sqrt{\frac{c\pi}{2}} \erf{\sqrt{\frac{c}{2}}} + \e^{-\frac{c}{2}} \right) - o(1).
    \end{align*}
    For the second case, $\varepsilon n \leq s \leq \frac{1}{4}n$, we use $\erf{\sqrt{\frac{c}{2}}\frac{s}{n}} \leq f_2(\sqrt{\frac{c}{2}}\frac{s}{n}) = \frac{2}{\sqrt{\pi}} \sqrt{\frac{c}{2}}\frac{s}{n}$ and $\e^{-\frac{c s^2}{2n^2}} \leq g_2(s) = 1 - \frac{c s^2}{2n^2} + \frac{c^2s^4}{8n^4}$. Before we plug that in Lemma \ref{lem:offlinesumtos1}, we combine all up to constant terms into $\O(1)$.
    \begin{align*}
        \sum_{k=0}^s \alpha_k &\geq \left(1 - \frac{c}{n}\right) \left( (s+1) \bigg( \sqrt{\frac{c\pi}{2}} \erf{\sqrt{\frac{c}{2}}} - \frac{c}{n} + \e^{-\frac{c}{2}} \right)\\
        &\quad\quad\quad\quad\quad\quad + n - \sqrt{\frac{\pi}{2c}} (n + c (s+1)) \erf{\sqrt{\frac{c}{2}}\frac{s}{n}} - (n+1) \e^{-\frac{c s^2}{2 n^2}} \bigg)\\
        &= s \left( \sqrt{\frac{c\pi}{2}} \erf{\sqrt{\frac{c}{2}}} + \e^{-\frac{c}{2}} \right) + n - \sqrt{\frac{\pi}{2c}} (n + c s) \erf{\sqrt{\frac{c}{2}}\frac{s}{n}} - n \cdot \e^{-\frac{c s^2}{2 n^2}} \pm \O(1)\\
        &\geq s \left( \sqrt{\frac{c\pi}{2}} \erf{\sqrt{\frac{c}{2}}} + \e^{-\frac{c}{2}} \right) + n - \sqrt{\frac{\pi}{2c}} (n + c s) \frac{2}{\sqrt{\pi}}\frac{\sqrt{c}}{\sqrt{2}}\frac{s}{n} - n \left(1 - \frac{c s^2}{2n^2} + \frac{c^2s^4}{8n^4}\right) \pm \O(1)\\
        &=  s \left( \sqrt{\frac{c\pi}{2}} \erf{\sqrt{\frac{c}{2}}} + \e^{-\frac{c}{2}} \right) - (n + c s) \frac{s}{n}  + n \left(\frac{c s^2}{2n^2} - \frac{c^2s^4}{8n^4}\right) \pm \O(1)\\
        &=  s \left( \sqrt{\frac{c\pi}{2}} \erf{\sqrt{\frac{c}{2}}} + \e^{-\frac{c}{2}} - 1 - \frac{c s}{n}\right) + \frac{c s^2}{2n} - \frac{c^2s^4}{8n^3} \pm \O(1).
    \end{align*}
    To get rid of the $s^4$, we apply $s \leq \frac{1}{4}n$ twice, yielding
    \begin{align*}
        \sum_{k=0}^s \alpha_k &\geq s \left( \sqrt{\frac{c\pi}{2}} \erf{\sqrt{\frac{c}{2}}} + \e^{-\frac{c}{2}} - 1 - \frac{c s}{n}\right) + \frac{c s^2}{2n} - \frac{c^2s^2}{128n} \pm \O(1)\\
        &= s \left( \sqrt{\frac{c\pi}{2}} \erf{\sqrt{\frac{c}{2}}} + \e^{-\frac{c}{2}} - 1 - \frac{c s}{n}\right) + s\frac{c s}{n} \left(\frac{64-c}{128}\right) \pm \O(1)\\
        &= s \left( \sqrt{\frac{c\pi}{2}} \erf{\sqrt{\frac{c}{2}}} + \e^{-\frac{c}{2}} - 1 - \frac{c s}{n} \frac{64+c}{128} \right) \pm \O(1).
    \end{align*}
    For the third case, $\frac{1}{4}n < s \leq \frac{1}{2}n$, we use $\erf{\sqrt{\frac{c}{2}}\frac{s}{n}} \leq f_3(\sqrt{\frac{c}{2}}\frac{s}{n}) = \frac{2\e^{-0.85^2}}{\sqrt{\pi}} \sqrt{\frac{c}{2}}\frac{s}{n} + \erf{0.85} - \frac{1.7 \e^{-0.85^2}}{\sqrt{\pi}}$ and $\e^{-\frac{c s^2}{2n^2}} \leq g_3(s) = -4 \e^{-\frac{c}{8}} \left(\e^{\frac{3c}{32}} - 1\right) \frac{s}{n} + \e^{-\frac{c}{8}} \left(2\e^{\frac{3c}{32}} - 1\right)$. Before we plug that in Lemma \ref{lem:offlinesumtos1}, recall that with combined up to constant terms, this yields
    \begin{align*}
        \sum_{k=0}^s \alpha_k &\geq s \left( \sqrt{\frac{c\pi}{2}} \erf{\sqrt{\frac{c}{2}}} + \e^{-\frac{c}{2}} \right) + n \left(1 - \e^{-\frac{c s^2}{2 n^2}}\right) - \sqrt{\frac{\pi}{2c}} (n + c s) \erf{\sqrt{\frac{c}{2}}\frac{s}{n}} \pm \O(1)\\
        &\geq s \left( \sqrt{\frac{c\pi}{2}} \erf{\sqrt{\frac{c}{2}}} + \e^{-\frac{c}{2}} \right) + n \left(1 + 4 \e^{-\frac{c}{8}} \left(\e^{\frac{3c}{32}} - 1\right) \frac{s}{n} - \e^{-\frac{c}{8}} \left(2\e^{\frac{3c}{32}} - 1\right)\right)\\
        &~~~ - \sqrt{\frac{\pi}{2c}} (n + c s) \left(\frac{2\e^{-0.85^2}}{\sqrt{\pi}} \sqrt{\frac{c}{2}}\frac{s}{n} + \erf{0.85} - \frac{1.7 \e^{-0.85^2}}{\sqrt{\pi}}\right) \pm \O(1)\\
        &= s \left( \sqrt{\frac{c\pi}{2}} \erf{\sqrt{\frac{c}{2}}} + \e^{-\frac{c}{2}} + 4 \e^{-\frac{c}{8}} \left(\e^{\frac{3c}{32}} - 1\right) - \left(1 + \frac{c s}{n}\right) \e^{-0.85^2} \right)\\
        &~~~ + n \left(1 - \e^{-\frac{c}{8}} \left(2\e^{\frac{3c}{32}} - 1\right)\right) - \sqrt{\frac{\pi}{2c}} (n + c s) \left(\erf{0.85} - \frac{1.7 \e^{-0.85^2}}{\sqrt{\pi}}\right) \pm \O(1).
    \end{align*}
    For the fourth case, $\frac{1}{2}n < s \leq \frac{3}{4}n$, we use $\erf{\sqrt{\frac{c}{2}}\frac{s}{n}} \leq f_4(\sqrt{\frac{c}{2}}\frac{s}{n}) = \frac{2\e^{-1.35^2}}{\sqrt{\pi}} \sqrt{\frac{c}{2}}\frac{s}{n} + \erf{1.35} - \frac{2.7 \e^{-1.35^2}}{\sqrt{\pi}}$ and $\e^{-\frac{c s^2}{2n^2}} \leq g_4(s) = -4 \e^{-\frac{9c}{32}} \left(\e^{\frac{5c}{32}} - 1\right) \frac{s}{n} + \e^{-\frac{9c}{32}} \left(3\e^{\frac{5c}{32}} - 2\right)$. The calculation is fully analogous to the previous case, yielding
    \begin{align*}
        \sum_{k=0}^s \alpha_k &\geq s \left( \sqrt{\frac{c\pi}{2}} \erf{\sqrt{\frac{c}{2}}} + \e^{-\frac{c}{2}} \right) + n \left(1 + 4 \e^{-\frac{9c}{32}} \left(\e^{\frac{5c}{32}} - 1\right) \frac{s}{n} - \e^{-\frac{9c}{32}} \left(3\e^{\frac{5c}{32}} - 2\right) \right)\\
        &~~~ - \sqrt{\frac{\pi}{2c}} (n + c s) \left(\frac{2\e^{-1.35^2}}{\sqrt{\pi}} \sqrt{\frac{c}{2}}\frac{s}{n} + \erf{1.35} - \frac{2.7 \e^{-1.35^2}}{\sqrt{\pi}}\right) \pm \O(1)\\
        &= s \left( \sqrt{\frac{c\pi}{2}} \erf{\sqrt{\frac{c}{2}}} + \e^{-\frac{c}{2}} + 4 \e^{-\frac{9c}{32}} \left(\e^{\frac{5c}{32}} - 1\right) - \left(1 + \frac{c s}{n}\right) \e^{-1.35^2} \right)\\
        &~~~ + n \left(1 - \e^{-\frac{9c}{32}} \left(3\e^{\frac{5c}{32}} - 2\right) \right) - \sqrt{\frac{\pi}{2c}} (n + c s) \left(\erf{1.35} - \frac{2.7 \e^{-1.35^2}}{\sqrt{\pi}}\right) \pm \O(1).
    \end{align*}
    Lastly, for the fifth case, $\frac{3}{4}n < s \leq n$, we use $\erf{\sqrt{\frac{c}{2}}\frac{s}{n}} \leq f_5(\sqrt{\frac{c}{2}}\frac{s}{n}) = 1$ and $\e^{-\frac{c s^2}{2n^2}} \leq g_5(s) = -4 \e^{-\frac{c}{2}} \left(\e^{\frac{7c}{32}} - 1\right) \frac{s}{n} + \e^{-\frac{c}{2}} \left(4\e^{\frac{7c}{32}} - 3\right)$. With analogous calculation, this yields
    \begin{align*}
        \sum_{k=0}^s \alpha_k &\geq s \left( \sqrt{\frac{c\pi}{2}} \erf{\sqrt{\frac{c}{2}}} + \e^{-\frac{c}{2}} \right) + n \left(1 + 4 \e^{-\frac{c}{2}} \left(\e^{\frac{7c}{32}} - 1\right) \frac{s}{n} - \e^{-\frac{c}{2}} \left(4\e^{\frac{7c}{32}} - 3\right) \right)\\
        &~~~ - \sqrt{\frac{\pi}{2c}} (n + c s) \cdot \erf{\sqrt{\frac{c}{2}}} \pm \O(1)\\
        &= s \left( \sqrt{\frac{c\pi}{2}} \erf{\sqrt{\frac{c}{2}}} + \e^{-\frac{c}{2}} + 4 \e^{-\frac{c}{2}} \left(\e^{\frac{7c}{32}} - 1\right) \right) + n \left(1 - \e^{-\frac{c}{2}} \left(4\e^{\frac{7c}{32}} - 3\right) \right)\\
        &~~~ - \sqrt{\frac{\pi}{2c}} (n + c s) \erf{\sqrt{\frac{c}{2}}} \pm \O(1),
    \end{align*}
    finishing the proof of the last case.
\end{proof}

Finally, we are ready to prove our main result. For the final proof of the theorem, we compare the different cases of Lemma \ref{lem:offlinesumtos2} and Lemma \ref{lem:onlinesumtos2} individually and, each time, take the limit for $n$ towards infinity. The minimum ratio over all the cases is our result. It is attained for the edge-cases $s=0$ and $s\in\{n-1,n\}$.

\begin{proof}[Proof of Theorem \ref{thm:onlbound}]
    Combining the results of Lemma \ref{lem:alphaoff} and Lemma \ref{lem:alphaon} with the bound from Lemma \ref{lem:minbound}, we get that
    \begin{equation}\label{eq:finalbound}
        \lim_{n\to\infty} \frac{\E[\textsc{Onl}_n]}{\E[\textsc{Opt}_n]} \geq \lim_{n\to\infty} \frac{\sum_{k=0}^n \alpha_k \E[x \mid \delta_{p(k+1)} < x \leq  \delta_{p(k)}]}{\sum_{k=0}^n \alpha_k^* \E[x \mid \delta_{p(k+1)} < x \leq \delta_{p(k)}]} \geq \lim_{n\to\infty} \min_{s\in \{0,\dots,n\}} \frac{\sum_{k=0}^{s} \alpha_k}{\sum_{k=0}^{s} \alpha^*_k},
    \end{equation}
    which we aim to lower bound by $0.598$. In order to find the minimum, we will go through all possible value ranges of $s$ that arose in Lemma \ref{lem:offlinesumtos2} and Lemma \ref{lem:onlinesumtos2}. Recall the constants $c=9.71$, $b_1 = \sqrt{\frac{c\pi}{2}} \erf{\sqrt{\frac{c}{2}}} + \e^{-\frac{c}{2}} - 1 \approx 2.906$ and $\varepsilon^\prime = c \varepsilon \left( \sqrt{\frac{c\pi}{2}} \erf{\sqrt{\frac{c}{2}}} + \e^{-\frac{c}{2}} \right) \approx 37.928\varepsilon > 0$.\\
    First, consider $s=0$. We get that
    \[\lim_{n\to\infty} \frac{\alpha_0}{\alpha_0^*} \geq \lim_{n\to\infty} \frac{(0+1)  b_1 - \varepsilon^\prime - o(1)}{\frac{c}{2} + o(1)} = \frac{b_1 - \varepsilon^\prime}{\frac{c}{2}},\]
    which for a very small choice of $\varepsilon$ is arbitrarily close to $\frac{2b_1}{c} \approx 0.598572 > 0.598$.\\
    Next, for $0 < s \leq \varepsilon n$, recall the constant $\varepsilon^{\prime\prime} = \varepsilon \frac{c(c+2)}{2} > 0$. Then,
    \begin{align*}\lim_{n\to\infty} \frac{\sum_{k=0}^s \alpha_k}{\sum_{k=0}^s \alpha_k^*} &\geq \lim_{n\to\infty} \frac{(s+1)  b_1 - \varepsilon^\prime - o(1)}{\frac{c(s+1)}{2} + \varepsilon^{\prime\prime} \pm o(1)} = \lim_{n\to\infty} \frac{b_1 - \frac{\varepsilon^\prime}{s+1} - o(1)}{\frac{c}{2} + \frac{\varepsilon^{\prime\prime}}{s+1} \pm o(1)}\\
    &\geq \lim_{n\to\infty} \frac{b_1 - \varepsilon^\prime - o(1)}{\frac{c}{2} + \varepsilon^{\prime\prime} \pm o(1)} = \frac{b_1 - \varepsilon^\prime}{\frac{c}{2} + \varepsilon^{\prime\prime}},
    \end{align*}
    which again for a very small choice of $\varepsilon$ is arbitrarily close to $\frac{2b_1}{c}\approx 0.598572$.\\
    Next, for $\varepsilon n \leq s \leq \frac{1}{4}n$, we get that
    \begin{align*}
        \lim_{n\to\infty} \frac{\sum_{k=0}^s \alpha_k}{\sum_{k=0}^s \alpha_k^*} &\geq \lim_{n\to\infty} \frac{s \left( b_1 - \frac{c s}{n} \frac{64+c}{128} \right) \pm \O(1)}{\frac{c}{2}s - \frac{64-16c+2c^2-64\e^{-\frac{c}{4}}}{n c}s^2 \pm \O(1)}\\
        &= \lim_{n\to\infty} \frac{b_1 - \frac{c(64+c)}{128} \frac{s}{n} \pm o(1)}{\frac{c}{2} - \frac{64-16c+2c^2-64\e^{-\frac{c}{4}}}{c}\frac{s}{n} \pm o(1)} = \frac{b_1 - \frac{c(64+c)}{128} a}{\frac{c}{2} - \frac{64-16c+2c^2-64\e^{-\frac{c}{4}}}{c} a},
    \end{align*}
    where we rewrote $s = a n$. We need to maximize this term as a function in $a \in [\varepsilon, \frac{1}{4}]$. To find extreme points, we use the quotient rule in order to calculate the first derivative. Note that for two linear functions, this always yields a constant value. It is easy to verify that this constant value is unequal to zero for our functions, such that the quotient has no extreme point. Also, the only zero of the denominator is for some $a>\frac{1}{2}$, such that there is no pole in the range of interest. Consequently, we only need to check the left and right boundaries, which are $a=\varepsilon$ and $a=\frac{1}{4}$, where choosing a small $\varepsilon$ again yields $\frac{2b_1}{c}\approx 0.598572$, and the second one equals
    \[\frac{b_1 - \frac{c(64+c)}{128} \frac{1}{4}}{\frac{c}{2} - \frac{64-16c+2c^2-64\e^{-\frac{c}{4}}}{c} \frac{1}{4}} \approx 0.603837 > 0.598.\]
    Next, for $\frac{1}{4} n \leq s \leq \frac{1}{2}n$, recall
    \[b_2 = \sqrt{\frac{c\pi}{2}} \erf{\sqrt{\frac{c}{2}}} + \e^{-\frac{c}{2}} + 4 \e^{-\frac{c}{8}} \left(\e^{\frac{3c}{32}} - 1\right) - \e^{-0.85^2} - c \sqrt{\frac{\pi}{2c}} \left(\erf{0.85} - \frac{1.7 \e^{-0.85^2}}{\sqrt{\pi}}\right) \approx 3.994\]
    and $b_3 = 1 - \e^{-\frac{c}{8}} \left(2\e^{\frac{3c}{32}} - 1\right) - \sqrt{\frac{\pi}{2c}} \left(\erf{0.85} - \frac{1.7 \e^{-0.85^2}}{\sqrt{\pi}}\right) \approx -0.302$. We get that
    \begin{align*}
        \lim_{n\to\infty} \frac{\sum_{k=0}^s \alpha_k}{\sum_{k=0}^s \alpha_k^*} &\geq \lim_{n\to\infty} \frac{s \left(b_2 - \frac{c s}{n} \e^{-0.85^2} \right) + n b_3 \pm \O(1)}{n - \frac{n^2}{c s} \left(1-\e^{-\frac{c}{4}}\right) \pm \O(1)}\\
        &= \lim_{n\to\infty} \frac{\frac{s}{n} \left(b_2 - \frac{c s}{n} \e^{-0.85^2} \right) + b_3 \pm o(1)}{1 - \frac{n}{c s} \left(1-\e^{-\frac{c}{4}}\right) \pm o(1)}\\
        &= \frac{a \left(b_2 - c a \e^{-0.85^2} \right) + b_3}{1 - \frac{1}{c a} \left(1-\e^{-\frac{c}{4}}\right)},
    \end{align*}
    where we again rewrote $s = a n$. We need to minimize this term as a function in $a \in [\frac{1}{4}, \frac{1}{2}]$. To find extreme points, we use the quotient rule in order to calculate the first derivative and set it equal to zero. Multiplying both sides with $a^2$, we can see that fact we search for the zeros of a cubic function, which has three zeros, of which only one is in the interval $[\frac{1}{4}, \frac{1}{2}]$, at approximately $a \approx 0.37264$. Plugging the exact value of $a$ into the original fraction yields a value of approximately $0.710687 > 0.598$. Also, the denominator of the fraction has exactly one zero outside of the interval $[\frac{1}{4}, \frac{1}{2}]$, such that there is no pole in the range of interest. Consequently, we only need to check the left and right boundaries, which are $a=\frac{1}{4}$ and $a=\frac{1}{2}$, where the first one yields a value of approximately $0.643427 > 0.598$ and the second one a value of approximately $0.635727 > 0.598$.\\
    Next, for $\frac{1}{2} n \leq s \leq \frac{3}{4}n$, recall
    \[b_4 = \sqrt{\frac{c\pi}{2}} \erf{\sqrt{\frac{c}{2}}} + \e^{-\frac{c}{2}} + 4 \e^{-\frac{9c}{32}} \left(\e^{\frac{5c}{32}} - 1\right) - \e^{-1.35^2} - c \sqrt{\frac{\pi}{2c}} \left(\erf{1.35} - \frac{2.7 \e^{-1.35^2}}{\sqrt{\pi}}\right) \approx 1.948\]
    and $b_5 = 1 - \e^{-\frac{9c}{32}} \left(3\e^{\frac{5c}{32}} - 2\right) - \sqrt{\frac{\pi}{2c}} \Big(\erf{1.35} - \frac{2.7 \e^{-1.35^2}}{\sqrt{\pi}}\Big) \approx -0.041$. Since the lower bound in this case results from the same formula as in the previous case, only with $b_2$ switched to $b_4$ and $b_3$ switched to $b_5$, we can use analogous calculations. Here, we find that the quotient has one extreme point in the range of interest at approximately $a \approx 0.5558$ with a value of approximately $0.669426 > 0.598$, but no poles in the range of interest. The left boundary at $a=\frac{1}{2}$ has a value of $0.664978 > 0.598$ and the right boundary at $a=\frac{3}{4}$ has a value of $0.613443 > 0.598$.\\
    Next, for $\frac{1}{2} n \leq s \leq n-2$, recall the constants $b_6 = \sqrt{\frac{c\pi}{2}} \erf{\sqrt{\frac{c}{2}}} + \e^{-\frac{c}{2}} + 4 \e^{-\frac{c}{2}} \left(\e^{\frac{7c}{32}} - 1\right) - c \sqrt{\frac{\pi}{2c}} \erf{\sqrt{\frac{c}{2}}} \approx 0.237$ and $b_7 = 1 - \e^{-\frac{c}{2}} \left(4\e^{\frac{7c}{32}} - 3\right) - \sqrt{\frac{\pi}{2c}} \erf{\sqrt{\frac{c}{2}}} \approx 0.361$. We get that
    \begin{align*}
        \lim_{n\to\infty} \frac{\sum_{k=0}^s \alpha_k}{\sum_{k=0}^s \alpha_k^*} &\geq \lim_{n\to\infty} \frac{s b_6 + n b_7 \pm \O(1)}{n - \frac{n^2}{c s} \left(1-\e^{-\frac{c}{4}}\right) \pm \O(1)}\\
        &= \lim_{n\to\infty} \frac{\frac{s}{n} b_6 + b_7 \pm o(1)}{1 - \frac{n}{c s} \left(1-\e^{-\frac{c}{4}}\right) \pm o(1)} = \frac{a b_6 + b_7}{1 - \frac{1}{c a} \left(1-\e^{-\frac{c}{4}}\right)},
    \end{align*}
    where we again rewrote $s = a n$. We need to minimize this term as a function in $a \in [\frac{3}{4}, 1]$. To find extreme points, we use the quotient rule in order to calculate the first derivative and set it equal to zero. Multiplying both sides with $a^2$, we can see that fact we search for the zeros of a quadratic function, where we can use the p-q-formula in order to derive the two zeros which are both not in the range of interest. The unique pole is at approximately $a = 0.0938973$, which is also not in the range of interest. The left boundary at $a=\frac{3}{4}$ has a value of $0.616382 > 0.598$ and the right boundary at $a=1$ has a value of $0.660554 > 0.598$.\\
    Lastly, for $s\in\{n-1, n\}$, we get that
    \begin{align*}
        \lim_{n\to\infty} \frac{\sum_{k=0}^s \alpha_k}{\sum_{k=0}^s \alpha_k*} \geq \lim_{n\to\infty} \frac{s b_6 + n b_7 \pm \O(1)}{n \pm \O(1)} = \lim_{n\to\infty} \frac{\frac{s}{n} b_6 + b_7 \pm o(1)}{1 \pm o(1)} = b_6 + b_7 \approx 0.598529 > 0.598.
    \end{align*}
    Overall, we found that for every possible value of $s$, the fraction on the right hand side of Equation \eqref{eq:finalbound} is always greater than $0.598$ for $n\to\infty$, such that the corresponding minimum gives a lower bound on the prophet inequality of \textsc{Onl} of $0.598$.
\end{proof}

\section{Conclusion}
We analyzed an over-time variant of the well-known prophet inequality with i.i.d.\ random variables. We presented an algorithm resulting in a prophet inequality of at least $\approx 0.598$ when the number of steps tends to infinity. Our upper bound shows that the best possible prophet inequality is at most $1/\varphi \approx 0.618$, where $\varphi$ denotes the golden ratio.

It is easy to see that our procedure can also be seen as a posted-price mechanism approximating social welfare in an online auction. Thus, the algorithm can even be used without direct access to the realizations of the variables. To see this, suppose customers draw the value for the item from the given joint distribution. The mechanism is the following. In round $i$, whenever the item is available, customer $i$ can rent the item for the current round for free. For all later rounds, the price is given by the current threshold of the algorithm. This yields the following decisions by the customers. Independent of the drawn value, all customers will always accept the item for free for the current round. Additionally, they will rent the item for all remaining rounds if and only if the posted price is below the drawn valuation. Thus, the generated welfare is exactly that of our online procedure and the posted-price mechanism approximates social welfare by the achieved approximation ratio of $0.598$ for $n\to\infty$.      

Although our online procedure \textsc{Onl} from Section \ref{sec:betterlower} is close to optimal, there is a small gap between the upper and lower bound on the prophet inequality for $n\to\infty$.
A natural extension of the model for future work is to introduce different distributions instead of i.i.d.\ random variables.
%
% ---- Bibliography ----
%
\bibliographystyle{plain}
\bibliography{full_references}

\end{document}